\documentclass[11pt]{article} 

\usepackage{fullpage}

\usepackage{amsmath,amsthm, amssymb}
\usepackage{latexsym}
\usepackage{graphicx}
\usepackage{ifthen}
\usepackage{subcaption}
\usepackage{multicol}
\usepackage{algorithm, algorithmic}
\usepackage[numbers]{natbib}
\usepackage{dsfont}
\usepackage{nicefrac}
\usepackage{mathtools}

\usepackage{hyperref}
\hypersetup{colorlinks=true,linkcolor=blue,citecolor=blue,urlcolor=blue}

\usepackage{lipsum}
\usepackage{mwe}

\newtheorem*{theorem*}{Theorem}
\newtheorem*{definition*}{Definition}

\newtheorem{lemma}{Lemma}
\newtheorem{theorem}{Theorem}
\newtheorem{definition}{Definition}
\newtheorem{corollary}{Corollary}
\newtheorem{claim}{Claim}

\usepackage{thm-restate}

\usepackage{color}
\definecolor{mygray}{gray}{0.5}

\newcommand\iid{\stackrel{\mathclap{\tiny\mbox{i.i.d.}}}{ \ \sim \ }}
\renewcommand{\algorithmiccomment}[1]{\bgroup\hfill\footnotesize~#1\egroup}
\newcommand{\algone}{{\textsc{Adaptive Sequencing}}}
\newcommand{\algoneplus}{{\textsc{Adaptive Sequencing++}}}
\newcommand{\algcontinuous}{{\textsc{Accelerated Continuous Greedy}}}

\newcommand{\blackbox}{{\textsc{Random Sequence}}}
\newcommand{\rank}{{\textsc{rank}}}
\newcommand{\Span}{{\textsc{span}}}

\newcommand{\INPUT}{\item[{\bf Input:}]}

\DeclareMathOperator{\E}{\mathbb{E}}

\newcommand{\zeros}{{\mathbf 0}}
\newcommand{\ones}{{\mathbf 1}}

\DeclareMathOperator{\argmax}{argmax}

\DeclareMathOperator{\poly}{poly}
\newcommand{\OPT}{\texttt{OPT}}

\newcommand{\R}{\mathbb{R}}                     % Reals.
\newcommand{\M}{\mathcal{M}}                     % Matroids.
\newcommand{\U}{\mathcal{U}}

\renewcommand{\O}{\mathcal{O}}

\newcommand{\bx}{\mathbf{x}}

\newcommand{\by}{\mathbf{y}}

\title{An Optimal Approximation for Submodular Maximization under a Matroid Constraint in the Adaptive Complexity Model}

\author{Eric Balkanski\\Harvard University \\ ericbalkanski@g.harvard.edu 
\and Aviad Rubinstein\\ Stanford University \\ aviad@cs.stanford.edu
\and Yaron Singer\\ Harvard University \\ yaron@seas.harvard.edu\\}

\date{}

\begin{document}

\setcounter{page}{0}

\maketitle
\begin{abstract}
In this paper we study submodular maximization under a matroid constraint in the adaptive complexity model.  This model was recently introduced in the context of submodular optimization in~\cite{BS18a} to quantify the information theoretic complexity of black-box optimization in a parallel computation model.  Informally, the \emph{adaptivity} of an algorithm is the number of sequential rounds it makes when each round can execute polynomially-many function evaluations in parallel.  Since submodular optimization is regularly applied on large datasets we seek algorithms with low adaptivity to enable speedups via parallelization.   Consequently, a recent line of work has been devoted to designing constant factor approximation algorithms for maximizing submodular functions under various constraints in the adaptive complexity model~\cite{BS18a,BS18b, BBS18,BRS19, EN19,FMZ19,chekuri2018submodular,ene2018submodular,FMZ18}.  

Despite the burst in work on submodular maximization in the adaptive complexity model, the fundamental problem of maximizing a monotone submodular function under a matroid constraint has remained elusive.  In particular, all known techniques fail for this problem and there are no known constant factor approximation algorithms whose adaptivity is sublinear in the rank of the matroid $k$ or in the worst case sublinear in the size of the ground set $n$. 

In this paper we present an approximation algorithm for the problem of maximizing a monotone submodular function under a matroid constraint in the adaptive complexity model.  The approximation guarantee of the algorithm is arbitrarily close to the optimal $1-1/e$ and it has near optimal adaptivity of $\O(\log(n)\log(k))$.  This result is obtained using a novel technique of \emph{adaptive sequencing} which departs from previous techniques for submodular maximization in the adaptive complexity model.  In addition to our main result we show how to use this technique to design other approximation algorithms with strong approximation guarantees and polylogarithmic adaptivity. 
\end{abstract}

\newpage
% !TeX root = main.tex

\section{Introduction}
In this paper we study submodular maximization under matroid constraints in the adaptive complexity model.  The adaptive complexity model was recently introduced in the context of submodular optimization in~\cite{BS18a} to quantify the information theoretic complexity of black-box optimization in a parallel computation model.  Informally, the \emph{adaptivity} of an algorithm is the number of sequential rounds it makes when each round can execute polynomially-many function evaluations in parallel.  The concept of adaptivity is heavily studied in computer science and optimization as it provides a measure of efficiency of parallel computation.

Since submodular optimization is regularly applied on very large datasets, we seek algorithms with low adaptivity to enable speedups via parallelization.  For the basic problem of maximizing a monotone submodular function under a cardinality constraint $k$ the celebrated greedy algorithm which iteratively adds to the solution the element with largest marginal contribution  is $\Omega(k)$ adaptive.  Until very recently, even for this basic problem, there was no known constant-factor approximation algorithm whose adaptivity is sublinear in $k$.  In the worst case $k \in \Omega(n)$ and hence  greedy and all other algorithms  had adaptivity that is \emph{linear} in the size of the ground set.

The main result in~\cite{BS18a} is an \emph{adaptive sampling} algorithm for maximizing a monotone submodular function under a cardinality constraint that achieves a constant factor approximation arbitrarily close to $1/3$ in $\mathcal{O}(\log n)$ adaptive rounds as well as a lower bound that shows that no algorithm can achieve a constant factor approximation in $\tilde{o}(\log n)$ rounds.  Consequently, this algorithm provided a constant factor approximation with an exponential speedup in parallel runtime for monotone submodular maximization under a cardinality constraint.  

In~\cite{BRS19, EN19}, the adaptive sampling technique was extended to achieve an approximation guarantee arbitrarily close to the optimal $1-1/e$ in $\O(\log n)$ adaptive rounds. This result was then obtained with a linear number of queries \cite{FMZ19}, which is optimal. Functions with bounded curvature   have also been studied  using  adaptive sampling under a cardinality constraint~\cite{BS18b}. The more general family of  packing constraints, which includes partition and laminar matroids, has been considered in~\cite{chekuri2018submodular}. In particular, under $m$ packing constraints, a $1-1/e-\epsilon$ approximation was obtained in $\O(\log^2m \log n)$ rounds using a combination of continuous optimization  and multiplicative weight update techniques.

\subsection{Submodular maximization under a matroid constraint}
For the fundamental problem of maximizing a monotone submodular function under a general matroid constraint it is well known since the late 70s that the greedy algorithm  achieves a $1/2$ approximation~\cite{NWF78} and that even for the special case of cardinality constraint no algorithm can obtain an approximation guarantee better than $1-1/e$ using polynomially-many value queries~\cite{nemhauser1978best}.  Thirty years later, in seminal work, Vondr{\'{a}}k introduced the continuous greedy algorithm which approximately maximizes the multilinear extension of the submodular function~\cite{CCPV07} and showed it obtains the optimal $1-1/e$ approximation guarantee~\cite{vondrak08}.

Despite the surge of interest in adaptivity of submodular maximization, the problem of maximizing a monotone submodular function under a matroid constraint in the adaptive complexity model has remained elusive.  As we discuss in Section~\ref{related_work}, when it comes to matroid constraints there are fundamental limitations of the techniques developed in this line of work.  The best known adaptivity for obtaining a constant factor approximation guarantee for maximizing a monotone submodular function under a matroid constraint is achieved by the greedy algorithm and is linear in the rank of the matroid.  The best known adaptivity for obtaining the optimal $1-1/e$ guarantee is achieved by the continuous greedy and is linear in the size of the ground set.

\begin{center}
\emph{Is there an algorithm whose adaptivity is sublinear in the size of the rank of the matroid that obtains a constant factor approximation guarantee?}
\end{center}

\subsection{Main result}
Our main result is an algorithm for the problem of maximizing a monotone submodular function under a matroid constraint whose approximation guarantee is arbitrarily close to the optimal $1-1/e$ and has near optimal adaptivity of $\O(\log(n)\log(k))$.  
\begin{theorem*}
For any $\epsilon>0$ there is an $\O\left(\log(n)  \log\left(\frac{k}{\epsilon^3}\right)   \frac{1}{\epsilon^3}\right)$ adaptive algorithm that, with probability $1 - o(1)$, obtains a $1-1/e - \O(\epsilon)$ approximation for maximizing a monotone submodular function under a matroid constraint.  
\end{theorem*}
Our result provides an exponential improvement in the adaptivity for maximizing a monotone submodular function under a matroid constraint with an arbitrarily small loss in approximation guarantee.  As we later discuss, beyond the information theoretic consequences, this implies that a very broad class of combinatorial optimization problems can be solved exponentially faster in standard parallel computation models given appropriate representations of the matroid constraints. 

Our main result is largely powered by a new technique developed in this paper which we call \emph{adaptive sequencing}.  This technique proves to be extremely powerful and is a departure from all previous techniques for submodular maximization in the adaptive complexity model.  In addition to our main result we show that this technique gives us a set of other strong results that include: 

\begin{itemize}
\item An $\O(\log(n)\log(k))$ adaptive combinatorial algorithm that obtains a  $\frac{1}{2}-\epsilon$ approximation for monotone submodular maximization under a matroid constraint (Theorem~\ref{thm:comb});
\item An $\O(\log(n)\log(k))$ adaptive combinatorial algorithm that obtains a $\frac{1}{P+1} - \epsilon$ approximation for monotone submodular maximization under intersection of $P$ matroids (Theorem~\ref{thm:intersection});
\item An  $\O(\log(n)\log(k))$ adaptive algorithm that obtains an approximation of $1-1/e-\epsilon$ for monotone submodular maximization under a partition matroid constraint that can be implemented in the PRAM model with polylogarithmic depth (Appendix~\ref{sec:explicit}).
\end{itemize}

In addition to these results the adaptive sequencing technique can be used to design algorithms that achieve the same results as those for cardinality constraint in~\cite{BRS19,EN19,FMZ19} and for non-monotone submodular maximization under cardinality constraint as in~\cite{BBS18} (Appendix~\ref{sec:explicit}).

\subsection{Technical overview}
\label{sec:technical_overview}

The standard approach to obtain an approximation guarantee arbitrarily close to $1-1/e$ for maximizing a submodular function under a matroid constraint $\M$ is by the continuous greedy algorithm due to Vondr{\'{a}}k~\cite{vondrak08}.  This algorithm approximately maximizes the multilinear extension $F$ of the submodular function~\cite{CCPV07} in $\O(n)$ adaptive steps.  In each step the algorithm updates a continuous solution $\bx \in[0,1]$ in the direction of $\ones_S$, where $S$ is chosen by maximizing an additive function under a matroid constraint.  

In this paper we introduce the \emph{accelerated continuous greedy} algorithm whose approximation is arbitrarily close to the optimal $1-1/e$.  Similarly to continuous greedy, this algorithm approximately maximizes the multilinear extension by carefully choosing $S\in \M$ and updating the solution in the direction of $\ones_S$.  In sharp contrast to continuous greedy, however, the choice of $S$ is done in a manner that allows making a \emph{constant} number of updates to the solution, each requiring $\O(\log(n)\log(k))$ adaptive rounds.
We do this by constructing a feasible set $S$ using $\O(\log(n)\log(k))$ adaptive rounds, at each one of the $1/\lambda$ iterations of  accelerated continuous greedy, s.t. $S$ approximately maximizes the contribution of taking a step of constant size $\lambda$ in the direction of $\ones_{S}$.  We construct $S$   via a novel combinatorial algorithm introduced in Section~\ref{sec:comb}.  
 
The new combinatorial algorithm achieves by itself a $1/2$ approximation for submodular maximization under a matroid constraint in $\O(\log(n)\log(k))$ adaptive rounds.  This algorithm is developed using a fundamentally different approach from all previous low adaptivity algorithms for submodular maximization (see discussion in Section~\ref{related_work}).  This new framework uses a single random \emph{sequence} $(a_1, \ldots, a_k)$ of elements.  In particular, for each $i \in [k]$, element $a_i$ is chosen uniformly at random among all elements such that $S \cup \{a_1, \ldots, a_i\} \in \M$. This random feasibility of each element is central to the analysis. Informally, this ordering  allows the sequence to navigate randomly through the matroid constraint.   For each position $i$ in this sequence, we  analyze the number of elements $a$ such that $S \cup \{a_1, \ldots, a_i\} \cup a \in \M$ and $f_{S \cup \{a_1, \ldots, a_i\}}(a)$ is large. The key observation is that if this number is large at a position $i$, by the randomness of the sequence, $f_{S \cup \{a_1, \ldots, a_i\}}(a_{i+1})$ is large w.h.p., which is important for the approximation. Otherwise, if this number is low we discard a large number of elements, which is important for the adaptivity.

In Section~\ref{sec:continuous} we analyze the approximation of the  accelerated continuous greedy algorithm, which is the main result of the paper.  We use the algorithm from Section~\ref{sec:comb} to selects $S$ as the direction and show $F(\bx + \lambda \ones_S) - F(\bx) \geq (1 - \epsilon)\lambda(\OPT - F(\bx))$, which implies a $1-1/e-\epsilon$ approximation.  %We do so by optimizing $g(S) := F(\bx + \lambda \ones_S) - F(\bx)$ 
 
Finally, in Section~\ref{sec:matroid} we parallelize the matroid oracle queries.  The random sequence generated in each iteration of the combinatorial algorithm in Section~\ref{sec:comb} is independent of function evaluations and requires zero adaptive rounds, though it sequentially queries the matroid oracle.  For practical implementation it is important to parallelize the matroid queries to achieve fast parallel runtime.  When given explicit matroid constraints such as for uniform or partition matroids, this parallelization is relatively simple (Section~\ref{sec:explicit}). For general matroid constraints given via rank or independence oracles we show how to parallelize the matroid queries in Section~\ref{sec:matroid}.  We give upper and lower bounds by building on the seminal work of Karp, Upfal, and Wigderson on the parallel complexity of finding the base of a matroid~\cite{karp1988complexity}.  For rank oracles we show how to execute the algorithms with $\O(\log(n)\log(k))$ parallel steps that matches the $\O(\log(n)\log(k))$ adaptivity.  For independence oracles we show how to execute the algorithm using $\tilde{\O}(n^{1/2})$ steps of parallel matroid queries  and give an $\tilde{\Omega}(n^{1/3})$ lower bound even for additive functions and partition matroids.
% !TeX root = main.tex

\subsection{Previous optimization techniques in the adaptive complexity model}\label{related_work}

The random sequencing approach developed in this paper is a fundamental departure from the adaptive sampling approach introduced in~\cite{BS18a} and employed in previous combinatorial algorithms that achieve low adaptivity for submodular maximization~\cite{BS18b,BBS18,BRS19,EN19,FMZ19,FMZ18}. In adaptive sampling an algorithm samples multiple large feasible sets  at every iteration to determine elements which should be added to the solution or discarded. The issue with these uniformly  random feasible sets  is that, although they have a simple structure  for uniform matroids, they are complex objects to generate and analyze for general matroid constraints.

Chekuri and Quanrud recently obtained a $1 - 1/e - \epsilon$ approximation in $\O(\log^2m \log n)$ adaptive rounds for the family of $m$ packing constraints, which includes partition and laminar matroids~\cite{chekuri2018submodular}. This setting was then also considered for non-monotone functions  in~\cite{ene2018submodular}. Their approach also uses the continuous greedy algorithm, combined with a multiplicative weight update technique to handle the constraints. Since general matroids consist of exponentially many  constraints, a multiplicative weight update approach over these constraints is not feasible.  More generally packing constraints assume an explicit representation of the matroid. For general matroid constraints, the algorithm is not given such a representation but an oracle.  Access to an independence oracle for a matroid breaks these results as shown in Section~\ref{sec:matroid}: any constant factor approximation algorithm with an independence oracle must have  $\tilde{\Omega}(n^{1/3})$ sequential steps.

\subsection{Preliminaries}

\paragraph{Submodularity.}  A function $f : 2^N \rightarrow \R_+$ over ground set $N = [n]$ is \emph{submodular} if
the marginal contributions $f_S(a) := f(S \cup a) - f(S)$ of an element $a \in N\setminus S$ to a set $S \subseteq N$ are diminishing, meaning $f_S(a) \geq f_T(a)$ for all $S \subseteq T \subseteq N$ and $a \in N \setminus T$.  Throughout the paper, we abuse notation by writing $S \cup a$ instead of $S \cup \{a\}$ and  assume $f$ is monotone, so $f(S) \leq f(T)$ for all $S \subseteq T$.  The value of the optimal solution $O$ for the problem of maximizing the submodular function under some constraint $\M$ is denoted by $\OPT$, i.e. $O := \argmax_{S \in \M}f(S)$ and $\OPT := f(O)$.

\paragraph{Adaptivity.}  %Informally, the \emph{adaptivity} of an algorithm is the number of sequential rounds it makes when polynomially-many function evaluations can be executed in parallel in each round. 
%\begin{definition*}[Adaptivity \cite{BS18a}]
Given a value oracle for $f$, an algorithm  is \textbf{$r$-adaptive}
 if every query $f(S)$ for the value of a set $S$  occurs at a round $i \in [r]$ s.t. $S$ is independent of the values $f(S')$ of all other queries at round $i$, with at most $\poly(n)$ queries at every round.
% \end{definition*} 

\paragraph{Matroids.}  A set system $\M \subseteq 2^N$ is a \emph{matroid} if it satisfies the \emph{downward closed} and \emph{augmentation} properties. A set system $\M$ is downward closed if for all $S \subseteq T$ such that $T \in \M$, then $S \in \M$. The augmentation property is that if $S, T \in \M$ and $|S| < |T|$, then there exists $a \in T$ such that $S \cup a \in \M$. We call a set $S \in \M$ \emph{feasible} or \emph{independent}. The rank $k = \rank(\M)$ of a matroid is the maximum size of an independent set $S$. The rank $\rank(S)$ of a set $S$ is the maximum size of an independent subset $T \subseteq S$. A set $B \in \M$ is called a base of $\M$ if $|B| = \rank(\M)$. The matroid polytope $P(\M)$  is the collection of points $\bx \in [0,1]^n$ in the convex hull of the independent sets of $\M$, or equivalently the points $\bx$ such that $\sum_{i \in S} x_i \leq \rank(S)$ for all $S \subseteq [n]$.  

\paragraph{The multilinear extension.} The multilinear extension $F : [0,1]^n \rightarrow \R_+$ of a function $f$ maps a point $\bx \in [0,1]^n$ to  the expected value of a random set $R \sim \bx$ containing  each element $i \in [n]$ with probability $x_i$ independently, i.e. $F(\bx) = \E_{R \sim \bx}[f(R)]$. We note that given an oracle for $f$, one can estimate $F(\bx)$ arbitrarily well in one  round by querying in parallel a sufficiently large number of samples $R_1, \ldots, R_m \iid \bx$  and taking the average value of  $f(R_i)$ over $i \in [m]$ \cite{chekuri2015multiplicative,chekuri2018submodular}. For ease of presentation, we assume throughout the paper that we are given access to an exact value oracle for $F$ in addition to $f$. The results which rely on $F$ then extend to the case where the algorithm is only given an oracle for $f$ with an arbitrarily small loss in the approximation, no loss in the adaptivity, and additional $\O(n \log n)$ factor in the query complexity.\footnote{With $\O(\epsilon^{-2} n \log n)$ samples,   $F(\bx)$ is estimated within a $(1 \pm \epsilon)$ multiplicative factor with high probability\cite{chekuri2018submodular}.}

\section{The Combinatorial Algorithm}\label{sec:comb}
\label{sec:mainresult}

In this section we describe a combinatorial algorithm used at every iteration of the accelerated continuous greedy algorithm to find a direction $\ones_S$ for an update of a continuous solution.  In the next section we will show how to use this algorithm as a subprocedure in the accelerated continuous greedy algorithm to achieve an approximation arbitrarily close to $1-1/e$ with $\O(\log(n)\log(k))$ adaptivity.  The optimization of this direction $S$ is itself 
an instance of maximizing a monotone submodular function under a matroid constraint. The main result of this section  is  a $\O(\log(n)\log(k))$ adaptive algorithm, which we call  \algone, that returns a solution $\{a_i\}_i$ s.t., for all $i$,   the marginal contribution of $a_i$ to $\{a_1, \ldots, a_{i-1}\}$ is near optimal with respect to all elements $a$ s.t. $\{a_1, \ldots, a_{i-1}, a\} \in \M$. We note that this guarantee also implies that \algone \ itself achieves an approximation that is arbitrarily close to $1/2$ with high probability. 

As discussed in Section~\ref{sec:technical_overview} unlike all previous low-adaptivity combinatorial algorithms for submodular maximization, the \algone \ algorithm developed here does not iteratively sample large sets of elements in parallel at every iteration.  Instead,  it samples a \emph{single} random \emph{sequence} of elements in every iteration.  Importantly, this sequence is generated without any function evaluations, and therefore can be executed in zero adaptive rounds.  The goal is then to identify a high-valued prefix of the sequence that can be added to the solution and discard a large number of low-valued elements at every iteration.  Identifying a high valued prefix enables the approximation guarantee and discarding a large number of elements in every iteration ensures low adaptivity.

%Next, even though the order of the sequence is independent of $f$, we wish to both add a high-valued prefix of this sequence  to the solution  and also discard a large number of low-valued elements.   In particular, the algorithm  finds the smallest position in the sequence  such that a constant fraction of elements in $X$ are bad elements. 

\subsection{Generating random feasible sequences}

The algorithm crucially requires generating a random sequence of elements in zero adaptive rounds. %of adaptivity.  %Given solution $S$ we require the sequence $(a_1,\ldots,a_i)$ to respect that $S \cup \{a_1, \ldots, a_i\} \in \M$ for all $i$. Second we require $a_i$ to be uniformly random over all such elements which maintain feasibility. Note that this is condition independent of any function evaluations. We formalize these two properties as the random feasible sequence condition.

\begin{definition}
Given a matroid $\M$ we say that $(a_1, \ldots, a_{\rank(\M)})$ is a \textbf{random feasible sequence} if for all $i \in [\rank(\M)]$, $a_i$ is an element chosen u.a.r. from $\{a : \{a_1, \ldots, a_{i-1}, a\}  \in \M\}$.
\end{definition}

A simple way to obtain a random feasible sequence is by sampling feasible elements sequentially. %$a_1, \ldots, a_i$, as formally shown  in Algorithm~\ref{alg:bbs}. 

\begin{algorithm}[H]
\caption{ \blackbox \  }
\begin{algorithmic}
 	\INPUT matroid $\M$
	\STATE \textbf{for} $i = 1$ to $\rank(\M)$ \textbf{do}
	\STATE \ \ \ $X \leftarrow \{a : \{a_1, \ldots, a_{i-1}, a \} \in \M\}$
	\STATE \ \ \ $a_i \sim$ a uniformly random element from $X$
	\RETURN $a_1, \ldots, a_{\rank(\M)}$
  \end{algorithmic}
  \label{alg:bbs}
\end{algorithm}

It is immediate that Algorithm~\ref{alg:bbs} outputs a random feasible sequence. Since Algorithm~\ref{alg:bbs} is independent of $f$, its adaptivity is zero.  For ease of presentation, we describe the algorithm using \blackbox \ as a subroutine, despite its sequential calls to the matroid oracle.  In Section~\ref{sec:matroid} we show how to efficiently parallelize this procedure using standard matroid oracles.

\subsection{The algorithm}
The main idea behind the algorithm is to generate a random feasible sequence in each adaptive round, and use that sequence to determine which elements should be added to the solution and which should be discarded from consideration.  Given a position $i \in\{1,\ldots, l\}$ in a sequence $(a_1,a_2,\ldots,a_{l})$, a subset $S$, and threshold $t$, we say that an element $a$ is \emph{good} if adding it to $S \cup \{a_1, \ldots, a_{i}\}$ satisfies the matroid constraint and its marginal contribution to $S \cup \{a_1, \ldots, a_{i}\}$ is at least threshold $t$. 
In each adaptive round the algorithm generates a random feasible sequence and finds the index $i^\star$ which is the minimal index $i$ such that at most a $1 - \epsilon$ fraction of the surviving elements $X$  are good.  The algorithm then adds the set $\{a_1,\ldots,a_{i^\star}\}$ to $S$.  A formal description of the algorithm is included below.  We use $\M(S, X) := \{T \subseteq X : S \cup T \in \M\}$  to denote the matroid over elements $X$ where a subset is feasible in $\M(X,S)$ if its union with the current solution $S$ is feasible according to $\M$.

\begin{algorithm}[H]
\caption{\algone}
\begin{algorithmic}
    	\INPUT function $f$, feasibility constraint $\M$
    	\STATE  $S \leftarrow \emptyset, t \leftarrow \max_{a \in N}f(a)$
    	\STATE \textbf{for} $\Delta$ iterations \textbf{do}
	\STATE \ \  \ $X \leftarrow N$
	\STATE \ \ \  \textbf{while} $X \neq \emptyset$ \textbf{do}
	\STATE \ \ \ \ \ \ $a_1, \ldots, a_{\rank(\M(S,X))} \leftarrow \blackbox(\M(S,X))$
	\STATE \ \ \ \ \ \  $X_i \leftarrow \{a \in X :  S \cup \{a_1, \ldots, a_{i}, a\} \in \M \text{ and } f_{S \cup \{a_1, \ldots, a_{i}\}}(a) \geq  t\}$
	\STATE \ \ \  \ \ \ $i^{\star} \leftarrow \min\left\{i : |X_i| \leq (1 - \epsilon)|X|\right\}$ 
	\STATE \ \ \ \ \ \  $S \leftarrow S \cup \{a_1, \ldots, a_{i^{\star}}\}$
	\STATE \ \ \ \ \ \ $X \leftarrow  X_{i^{\star}}$
	\STATE \ \ \ $t \leftarrow (1 - \epsilon) t$
    	\RETURN $S$ 
  \end{algorithmic}
  \label{alg:1}
\end{algorithm}

Intuitively, adding $\{a_1,\ldots,a_{i^\star}\}$ to the current solution $S$ is desirable for two important reasons. First, for a random feasible sequence we have that $S \cup \{a_1,\ldots,a_{i^\star}\} \in \M$ and for each element $a_i$ at a  position $i \leq i^{\star}$, there is a high likelihood that the marginal contribution of $a_i$ to the previous elements in the sequence is at least $t$. Second, by definition of $i^\star$ a constant fraction $\epsilon$ of elements are not good at that position, and we discard these elements from $X$. This discarding guarantees that there are at most logarithmically many iterations until $X$ is empty.

The threshold $t$ maintains the invariant that it is approximately an upper bound on the optimal marginal contribution  to the current solution.  By submodularity, the optimal marginal contribution to $S$ decreases as $S$ grows. Thus, to maintain the invariant, the algorithm iterates over decreasing values of $t$. In particular,  at each of $\Delta = \O\left(\frac{1}{\epsilon}\log\left(\frac{k}{\epsilon}\right)\right)$ iterations, where $k := \rank(\M)$, the algorithm decreases $t$ by a $1-\epsilon$ factor when there are no more elements which can be added to $S$ with marginal contribution at least $t$, so when $X$ is empty. %Thus, at any iteration of the algorithm, . %The algorithm iterates over $\Delta = \O(\log k)$ decreasing values of $t$.

\subsection{Adaptivity} 
In each inner-iteration the algorithm  makes polynomially-many queries that are independent of each other.  Indeed, in each iteration, we generate $X_1,\ldots,X_{k - |S|}$ non-adaptively and make at most $n$ function evaluations for each $X_i$.  The adaptivity immediately follows from the definition of $i^\star$ that ensures an $\epsilon$ fraction of surviving elements in $X$ are discarded at every iteration.

\begin{lemma}
\label{lem:adaptivity} With $\Delta = \O\left(\frac{1}{\epsilon}\log\left(\frac{k}{\epsilon}\right)\right)$,
 \algone \ has adaptivity $\O\left(\log(n)\log\left(\frac{k}{\epsilon}\right)\frac{1}{\epsilon^2}\right)$.
\end{lemma}
\begin{proof}
The for loop has  $\Delta$ iterations. The  while loop has at most $\O(\epsilon^{-1} \log n)$ iterations since, by definition of $i^{\star}$, an $\epsilon$ fraction of the surviving elements are discarded from $X$ at every iteration.
We can find $i^{\star}$ by computing $X_i$ for each $i \in [k]$ in parallel in one round. 
\end{proof}

We note that the query complexity of the algorithm is $\O\left(nk\log(n)\log\left(\frac{k}{\epsilon}\right)\frac{1}{\epsilon^2}\right)$ and can be improved to $\O\left(n\log(n)\log(k)\log\left(\frac{k}{\epsilon}\right)\frac{1}{\epsilon^2}\right)$  if we allow $\O\left(\log(n)\log(k)\log\left(\frac{k}{\epsilon}\right)\frac{1}{\epsilon^2}\right)$ adaptivity by doing a binary search over at most $k$ sets $X_i$ to find $i^{\star}$.  The details can be found in Appendix~\ref{sec:appcomb}.

%\begin{lemma}
%The query complexity of Algorithm~\ref{alg:1} is $\O(nk \log(n) \log(k))$. This query complexity can be improved to  $n \log(n) \log^2(k)$ with adaptivity  $\O(\log(n) \log^2(k))$.
%\end{lemma}
%\begin{proof}
%For each $i \leq k$, there is $|X| \leq n$ evaluations of marginal contributions, so at most $n k$ queries  per iteration. By Lemma~\ref{lem:adaptivity}, there are $O(\log(n) \log(k))$ iterations.
%
%We observe that this query complexity can be improved by finding $i^{\star}$ by binary search over $[k-|S|]$ since $|X_i|$ is decreasing in $i$ (Lemma~\ref{lem:size} in Appendix~\ref{sec:appcomb}). With binary search, finding $i^{\star}$ takes at most $\log k$ rounds and $n \log k$ queries.
%\end{proof}
%

\subsection{Approximation guarantee} 
The main result for the approximation guarantee is that the algorithm returns a solution $S = \{a_1, \ldots, a_l\}$ s.t. for all $i \leq l$, the marginal contribution obtained by $a_i$ to $\{a_1, \ldots, a_{i-1}\}$ is near optimal with respect to all elements $a$ such that $\{a_1, \ldots, a_{i-1}, a\} \in \M$.  To prove this we show that the threshold $t$ is an approximate upper bound on the maximum marginal contribution. 

\begin{lemma} 
\label{lem:invariants} Assume that $f$ is submodular and that $\M$ is downward closed. Then, at any iteration, $t \geq (1 - \epsilon) \max_{a : S \cup a \in \M} f_S(a).$
\end{lemma} 
\begin{proof} 
The claim initially holds by the initial definitions of $t = \max_{a \in N}f(a)$, $S = \emptyset$ and $X = N$. We show that this invariant is maintained through the algorithm when either $S$ or $t$ are updated.

First, assume that at some iteration of the algorithm we have $t \geq (1 - \epsilon) \max_{a : S \cup a \in \M} f_S(a)$ and that $S$ is updated to $S \cup \{a_1, \ldots, a_{i^{\star}}\}$. Then, for all $a$ such that $S \cup a \in \M$, 
$$f_{S \cup \{a_1, \ldots, a_{i^{\star}}\}}(a) \leq f_{S}(a) \leq t/(1- \epsilon)$$
where the first inequality is by submodularity and the second by the inductive hypothesis. Since $\{a : S \cup \{a_1, \ldots, a_{i^{\star}}\} \cup a \in \M\} \subseteq \{a : S \cup a \in \M\}$ by the downward closed property of $\M$, 
$$\max_{a : S \cup \{a_1, \ldots, a_{i^{\star}}\} \cup a \in \M} f_{S \cup \{a_1, \ldots, a_{i^{\star}}\}}(a) \leq \max_{a : S \cup a \in \M} f_{S \cup \{a_1, \ldots, a_{i^{\star}}\}}(a).$$
Thus, when $S$ is updated to $S \cup \{a_1, \ldots, a_{i^{\star}}\}$, we have $t \geq (1 - \epsilon) \max_{a : S \cup \{a_1, \ldots, a_{i^{\star}}\} \cup a \in \M} f_S(a)$.

Next, consider an iteration where $t$ is updated to  $t' = (1 - \epsilon)t$. By the algorithm,  $X = \emptyset$ at that iteration with current solution $S$. Thus, by the algorithm, for all $a \in N$, $a$ was discarded from $X$ at some previous iteration with current solution $S'$ s.t.  $S' \cup \{a_1, \ldots, a_{i^{\star}}\} \subseteq S$. Since $a$ was discarded, it is either the case that $S' \cup \{a_1, \ldots, a_{i^{\star}}\} \cup a \not \in \M$ or $f_{S \cup \{a_1, \ldots, a_{i^{\star}}\}}(a) < t$. If $S' \cup \{a_1, \ldots, a_{i^{\star}}\} \cup a \not \in \M$ then $S \cup a \not \in \M$ by the downward closed property of $\M$ and since $S' \cup \{a_1, \ldots, a_{i^{\star}}\} \subseteq S$. Otherwise, $f_{S' \cup \{a_1, \ldots, a_{i^{\star}}\}}(a) < t$ and by submodularity, $f_S(a) \leq f_{S' \cup \{a_1, \ldots, a_{i^{\star}}\}}(a) <   t = t' / (1-\epsilon)$. Thus, $\forall a \in N$ s.t. $S \cup a \in \M$,  $t' \geq (1 - \epsilon) f_S(a)$ and the invariant is maintained. 
\end{proof}

By exploiting the definition of $i^{\star}$ and the random feasible sequence property we show that Lemma~\ref{lem:invariants} implies that every element added to $S$ at some iteration $j$  has near-optimal expected marginal contribution to $S$. We define  $X_i^{\M} := \{a \in X: S \cup \{a_1, \ldots, a_i\}\cup a \in \M\}$.

\begin{lemma}
\label{lem:marg}
Assume that $a_1, \ldots, a_{\rank(\M(S,X))}$   is a random feasible sequence, then for all $i \leq i^{\star}$,
$$\E_{a_i}\left[f_{S \cup \{a_1, \ldots, a_{i-1}\}}(a_i)\right] \geq (1 - \epsilon)^2 \max_{a : S \cup \{a_1, \ldots, a_{i-1}\} \cup a \in \M} f_{S \cup \{a_1, \ldots, a_{i-1}\}}(a_i).$$
\end{lemma}
\begin{proof}
By the random feasibility condition, we have $a_i \sim \U(X_{i-1}^{\M})$. We get
$$\Pr_{a_i}\left[f_{S \cup \{a_1, \ldots, a_{i-1}\}}(a_i) \geq t\right] \cdot t
  = \frac{|X_{i-1}|}{|X_{i-1}^{\M}|} \cdot t   
 \geq  \frac{|X_{i-1}|}{|X|} \cdot t  
 \geq (1-\epsilon)(1-\epsilon)\max_{a : S \cup \{a_1, \ldots, a_{i-1}\} \cup a \in \M} f_{S_{i-1}}(a_i)$$
where the equality is by definition of $X_{i-1}$,  the first inequality since $X_{i-1}^{\M}  \subseteq X$, and the second since $i \leq i^{\star}$ and by Lemma~\ref{lem:invariants}. Finally, note that $\E\left[f_{S \cup \{a_1, \ldots, a_{i-1}\}}(a_i)\right]  \geq \Pr\left[f_{S \cup \{a_1, \ldots, a_{i-1}\}}(a_i) \geq t\right] \cdot t$.
\end{proof}

Next, we show that if every element $a_i$ in a solution $S = \{a_1, \ldots, a_k\}$ of size $k = \rank(\M)$ has near-optimal expected marginal contribution to $S_{i-1} := \{a_1, \ldots, a_{i-1}\}$, then we obtain an approximation arbitrarily close to $1/2$ in expectation.
\begin{lemma}
Assume that $S = \{a_1, \ldots, a_k\}$ such that 
$\E_{a_i}[f_{S_{i-1}}(a_i)] \geq (1 - \epsilon) \max_{a: S_{i-1} \cup a \in \M}f_{S_{i-1}}(a)$ where $S_i = \{a_1, \ldots, a_i\}$. Then, for a matroid constraint $\M$, we have $\E\left[f(S)\right] \geq (1/2 - \O(\epsilon))\OPT$.
\end{lemma}
\begin{proof}
Let $O = \{o_1, \ldots, o_k\}$ such that $\{a_1, \ldots, a_{i-1}, o_i\}$ is feasible for all $i$, which exists by the augmentation property of matroids. We get,
$$
\E[f(S)] = \sum_{i \in [k]}\E[f_{S_{i-1}}(a_i)] \geq (1-\epsilon)\sum_{i \in [k]}\E[f_{S_{i-1}}(o_i)] \geq (1-\epsilon)f_S(O) \geq (1-\epsilon)(\OPT - f(S)). \qedhere$$
\end{proof}

A corollary of the lemmas above is that \algone \ has $\O(\log(n)\log(k))$ adaptive rounds and provides an approximation that is arbitrarily close to $1/2$, in expectation.  To obtain this guarantee with high probability we can simply run parallel instances of the while-loop in the algorithm and include the elements obtained from the best instance. We also note that the solution $S$ returned by \algone \ might have size smaller than $\rank(\M)$, which causes an arbitrarily small loss for sufficiently large $\Delta$.  We give the full details in Appendix~\ref{sec:appcomb}.

\begin{restatable}{rThm}{thmcomb}
\label{thm:comb}
For any $\epsilon>0$, there is an $\O\left(\log(n)\log\left(\frac{k}{\epsilon}\right)\frac{1}{\epsilon^2}\right)$ adaptive algorithm that obtains a $1/2- \O(\epsilon)$ approximation with probability $1 - o(1)$ for maximizing a monotone submodular function under a matroid constraint.
\end{restatable}

%In the next section we will use \algone \ in the accelerated continuous greedy algorithm and use the guarantees we proved here to show that the accelerated continuous greedy obtains an approximation that is arbitrarily close to the optimal $1-1/e$ guarantee.  Before doing so, we note that a $1/2$ approximation is already a big deal.

In Appendix~\ref{sec:appcomb}, we generalize this result and obtain a $1/(P+1)- \O(\epsilon)$ approximation with high probability for the  intersection of $P$ matroids.

% !TeX root = main.tex
%\newpage
\section{The Accelerated Continuous Greedy Algorithm}
\label{sec:continuous}
 
In this section we describe the accelerated continuous greedy algorithm that achieves the main result of the paper.  This algorithm employs the combinatorial algorithm from the previous section to construct a continuous solution which approximately maximizes the multilinear relaxation $F$ of the  function $f$.  This algorithm requires $\O(\log(n)\log(k))$ adaptive rounds and it produces a continuous solution whose  approximation to the optimal solution is with high probability arbitrarily close to $1-1/e$.    Finally, since the solution is continuous and we seek a feasible \emph{discrete} solution, it requires rounding.  Fortunately, by using either dependent rounding~\cite{chekuri2009dependent} or contention resolution schemes~\cite{vondrak2011submodular} this can be done with an arbitrarily small loss in the approximation guarantee  without any function evaluations, and  hence without any additional adaptive rounds. 

\subsection{The algorithm}
The accelerated continuous greedy algorithm follows the same principle as the (standard) continuous greedy algorithm~\cite{vondrak08}: at every iteration, the  solution $\bx \in[0,1]^n$ moves in the direction of a feasible set $S \in \M$.  The crucial difference between the accelerated continuous greedy and the standard continuous greedy is in the choice of this set $S$ guiding the direction in which $\bx$ moves.  This difference allows the accelerated continuous greedy to terminate after a \emph{constant} number of iterations, each of which has $\O(\log(n)\log(k))$ adaptive rounds, in contrast to the continuous greedy which requires a linear number of iterations.  

To determine the direction in every iteration, the accelerated continuous greedy applies \algone \ on the surrogate function $g$ that measures the marginal contribution to  $\bx$ when taking a step of size $\lambda$ in the direction of $S$.  That is, $g(S) := F_\bx(\lambda S) = F(\bx + \lambda S) - F(\bx)$ where we abuse notation and write $\lambda S$ instead of $\lambda   \ones_S$ for $\lambda \in[0,1]$ and $S\subseteq N$.  Since $f$ is a monotone submodular function it is immediate that $g$ is monotone and submodular as well.    

\begin{algorithm}[H]
\caption{\algcontinuous}
\begin{algorithmic}
    	\INPUT  matroid $\M$, step size $\lambda$
    	\STATE  $\bx \leftarrow \zeros$
    	\STATE \textbf{for} $1/\lambda$ \text{iterations}  \textbf{do}
    	\STATE \ \ \ define $g:2^N \to \mathbb{R}$ to be $g(T) = F_{\bx}(\lambda T)$
	\STATE \ \  \ $S \leftarrow \algone(g, \M)$ 
	\STATE \ \ \ $\bx \leftarrow \bx + \lambda S$
    	\RETURN $\bx$ 
  \end{algorithmic}
  \label{alg:continuous}
\end{algorithm}

The analysis shows that in every one of the $1/\lambda$ iterations, \algone \  finds  $S$  such that the contribution of taking a step of size $\lambda$ in the direction of $S$ is approximately a $\lambda$ fraction of $\OPT - F(\bx)$.  For any $\lambda$ this is a sufficient condition for obtaining the $1 - 1/e - \epsilon$ guarantee.  

The reason why the standard continuous greedy cannot be implemented with a constant number of rounds $1/\lambda$ is that in every round it moves in the direction of $\ones_{S}$ for  $S := \argmax_{T \in \M} \sum_{a \in T} g(a)$.  When $\lambda$ is constant  $F_{\bx}(\lambda S)$ is arbitrarily low due to the potential overlap between high valued singletons (see Appendix~\ref{sec:applayerone}).  Selecting $S$ using \algone \ is the crucial part of the accelerated continuous greedy which allows implementing it in a constant number of iterations.  

\subsection{Analysis}
We start by giving a sufficient condition on \algone \ to obtain the ${1-1/e-\O(\epsilon)}$ approximation guarantee. The analysis is standard and the proof is deferred to Appendix~\ref{sec:applayerone}.

\begin{restatable}{rLem}{lemmetaoneExp}
\label{lem:metaoneExp}
For a given matroid $\M$ assume that $\algone$  outputs $S \in \M$ s.t.
$\E_S\left[F_{\bx}(\lambda S) \right] \geq (1 - \epsilon)  \lambda   (\OPT - F(\bx))$ at every iteration  of \algcontinuous .
Then \algcontinuous \ outputs $\bx \in P(\M)$ s.t.
$\E[F(\bx)] \geq \left(1 - 1/e - \epsilon\right) \OPT$.
\end{restatable}

For a set $S=\{a_1,\ldots,a_k\}$ we define $S_i := \{a_1,\ldots,a_i\}$ and $S_{j:k} := \{a_{j}, \ldots, a_k\}$.  We use this notation in the lemma below.  The lemma is folklore and proved in Appendix~\ref{sec:applayerone} for completeness.  

\begin{restatable}{rLem}{lemmatroidO}
\label{lem:matroidO}
Let $\M$ be a matroid, then for any feasible sets $S=\{a_1,\ldots,a_k\}$ and $O$ of size $k$, there exists an ordering of $O = \{o_1, \ldots, o_k\}$  where for all $i \in [k]$, $S_i \cup O_{i+1:k} \in \M$ and $S_i \cap O_{i+1:k} = \emptyset$.
\end{restatable}

The following lemma is key in our analysis.  We argue that unless the algorithm already constructed $S$ of sufficiently large value, the sum of the contributions of the optimal elements to $S$ is arbitrarily close to the desired $\lambda(\OPT- F(\bx))$.

\begin{lemma}
\label{lem:Oik}
 Assume that $g(S) \leq \lambda(\OPT - F(\bx))$, then 
$\sum_i g_{S \setminus O_{i:k}}(o_i) \geq \lambda (1 - \lambda) (\OPT - F(\bx)).$
\end{lemma}
\begin{proof}
We first lower bound this sum of marginal contribution of optimal elements with the contribution of the optimal solution to the current solution $\bx + \lambda S$ at the end of the iteration:
\begin{align*}
\sum_{i \in [k]} g_{S \setminus O_{i:k}}(o_i)  =  \sum_{i \in [k]}  F_{\bx + \lambda S \setminus O_{i:k}}(\lambda o_i)
 \geq   \sum_{i \in [k]}  F_{\bx  + O_{i-1}+ \lambda S}(\lambda o_i) 
 \geq \lambda  \sum_{i \in [k]}  F_{\bx  + O_{i-1} + \lambda S }(o_i) 
= \lambda   F_{\bx +  \lambda S}(O)   
\end{align*}
where the first inequality is by submodularity and the second by the multilinearity of $F$. In the standard analysis of greedy algorithms the optimal solution $O$ may overlap with the current solution. In the continuous algorithm, since the algorithm takes steps of size $\lambda$, we can bound the overlap between the solution at this iteration $\lambda S$ and the optimal solution: 
\begin{align*}
  F_{\bx +  \lambda S}(O) =   F_{\bx}(O + \lambda S) - F_{\bx}(\lambda S)
 \geq F_{\bx}(O)  - \lambda(\OPT - F(\bx))
 =  (1- \lambda) \left(\OPT - F(\bx)\right)
\end{align*}
the first inequality is by monotonicity and lemma assumption and the second by monotonicity.
\end{proof}

As shown in Lemma~\ref{lem:matroidO}, \algone \ picks elements $a_i$ with near-optimal marginal contributions. Together with Lemma~\ref{lem:Oik} we get the desired bound on the contribution of $\lambda S$ to $\bx$.

\begin{lemma}
\label{lem:maincontinuous}
Let $\Delta = \O\left(\frac{1}{\epsilon} \log\left(\frac{k}{\epsilon \lambda}\right) \right)$ and $\lambda = \O(\epsilon)$.  For any $\bx$ such that $F(\bx) < (1-1/e) \OPT$, the set $S$ returned by $\algone(g, \M)$ satisfies $\E\left[F_{\bx}(\lambda S)\right] \geq (1- \O(\epsilon)) \lambda(\OPT - F(\bx)).$
\end{lemma}
\begin{proof}
Initially, we have  $t_i < \OPT$. After $\Delta = \O\left(\frac{1}{\epsilon} \log\left(\frac{k}{\epsilon \lambda}\right)\right)$  iterations of the outer loop of \algone, we get $t_f = (1-\epsilon)^{\Delta} \OPT = \O\left(\frac{\epsilon \lambda\OPT}{k}\right)$. 
We begin by adding dummy elements to $S$ so that $|S| = k$, which enables pairwise comparisons between $S$ and $O$. In particular, we consider $S'$, which is $S$ together with $\rank(\M) - |S|$ dummy elements $a_{|S| + 1}, \ldots a_k$ such that, for any $\by$ and $\lambda$, $F_{\by}(\lambda a) = t_f$, which is the value of $t$ when \algone \ terminates. Thus, by Lemma~\ref{lem:invariants}, for dummy elements $a_i$, $g_{S_{i-1}}(a_i) = t_f \geq (1 - \epsilon) \max_{a : S_{i-1} \cup a \in \M} g_{S_{i-1}}(a)$. 

We will conclude the proof by showing that $S$ is a good approximation to $S'$. From Lemma~\ref{lem:marg} that the contribution of $a_i$ to $S_{i-1}$ approximates the optimal contribution to $S_{i-1}$: 
$$\E\left[F_{\bx}(\lambda S')\right]  = \sum_{i=1}^k \E\left[g_{S_{i-1}}(a_i)\right] \geq \sum_{i=1}^k (1 - \epsilon)^2 \max_{a : S_{i-1} \cup a \in \M} g_{S_{i-1}}(a_i). $$

By Lemma~\ref{lem:matroidO} and submodularity, we have $ \max_{a : S_{i-1} \cup a \in \M} g_{S_{i-1}}(a_i) \geq g_{S \setminus O_{i:k}}(o_i).$
By Lemma~\ref{lem:Oik}, we also have
$\sum_{i=1}^k g_{S \setminus O_{i:k}}(o_i) \geq \lambda (1 - \lambda) (\OPT - F(\bx))$. Combining the previous pieces, we obtain
 $$\E\left[F_{\bx}(\lambda S')\right] \geq (1 - \epsilon)^2\lambda (1 - \lambda) (\OPT - F(\bx)).$$
 We conclude by removing the value of dummy elements,
$$
\E\left[F_{\bx}(\lambda S)\right]   = \E\left[F_{\bx}(\lambda S') - F_{\bx + \lambda S}(\lambda (S'\setminus S)) \right]
 \geq   \E\left[F_{\bx}(\lambda S')\right] - k t_f  \geq  \E\left[F_{\bx}(\lambda S')\right] - \epsilon \lambda\OPT.$$
The lemma assumes that $F(\bx) < (1-1/e)\OPT$ and $\lambda = \O(\epsilon)$, so $\OPT \leq e(\OPT - F(\bx))$  and $ \epsilon \lambda\OPT =  \O(\epsilon)  \lambda(\OPT - F(\bx))$.  We conclude that
$\E\left[F_{\bx}(\lambda S)\right]  \geq \left(1 - \O(\epsilon)\right)\lambda(\OPT - F(\bx))$.
\end{proof}

The approximation guarantee of the \algcontinuous \ follows from lemmas~\ref{lem:maincontinuous} and~\ref{lem:metaoneExp}, and the adaptivity from Lemma~\ref{lem:adaptivity}.
We defer the proof to Appendix~\ref{sec:applayerone}.

\begin{restatable}{rThm}{thmexpectation}
\label{thm:expectation}
For any $\epsilon>0$ \algcontinuous \ makes $  \O\left(\log(n)  \log\left(\frac{k}{\epsilon^2}\right)   \frac{1}{\epsilon^2}\right)$ adaptive rounds and obtains a $1-1/e - \O(\epsilon)$ approximation in expectation for maximizing a monotone submodular function under a matroid constraint.  
\end{restatable}

The final step in our analysis shows that the guarantee of \algcontinuous \ holds not only in expectation but also with high probability. 
To do so we argue in the lemma below that if over all iterations $i$,   $F_{\bx}(\lambda S)$  is close \emph{on average over the rounds} to  $\lambda(\OPT - F(\bx))$, we obtain an approximation arbitrarily close to $1 - 1/e$ with high probability.  The proof is in Appendix~\ref{sec:applayerone}.
 
\begin{restatable}{rLem}{lemmetaone}
\label{lem:metaone}
Assume that $\algone$  outputs $S \in \M$ s.t.
$F_{\bx}(\lambda S)  \geq \alpha_i   \lambda   (\OPT - F(\bx))$ at every iteration $i$ of \algcontinuous \ and that $\lambda \sum_{i=1}^{\lambda^{-1}} \alpha_i \geq 1 -\epsilon$. 
Then \algcontinuous  outputs $\bx \in P(\M)$ s.t.
$F(\bx) \geq \left(1 -1/e - \epsilon\right) \OPT$.
\end{restatable}

The  approximation $\alpha_i$ obtained at iteration $i$  is $1 - \O(\epsilon)$ in expectation by Lemma~\ref{lem:maincontinuous}. Thus, by a simple concentration bound, w.h.p. it is close to $1 - \O(\epsilon)$ in average over all iterations. Together with Lemma~\ref{lem:metaone}, this implies the $1 -1/e - \epsilon$ approximation w.h.p.. The details are in Appendix~\ref{sec:applayerone}.

\begin{restatable}{rThm}{thmhp}
\label{thm:hp}
\algcontinuous \ is an $\O\left(\log(n)  \log\left(\frac{k}{\epsilon \lambda}\right)   \frac{1}{\epsilon \lambda}\right)$ adaptive algorithm that, with probability $1 - \delta$, obtains a $1-1/e - \O(\epsilon)$ approximation for maximizing a monotone submodular function under a matroid constaint, with step size $\lambda = \O\left(\epsilon^2 \log^{-1}\left(\frac{1}{\delta}\right)\right)$.  
\end{restatable}

% !TeX root = main.tex

\section{Parallelization of Matroid Oracle Queries}
\label{sec:matroid}
Throughout the paper we relied on \blackbox \ as a simple procedure to generate a random feasible sequence to achieve our $\O(\log(n) \log(k))$ adaptive algorithm with an  approximation arbitrarily close to $1-1/e$.  Although \blackbox \ has zero adaptivity, it makes $\rank(\M)$ sequential steps depending on   membership in the matroid to generate the sets $X_1,\ldots,X_{\rank(M)}$.  From a practical perspective, we may wish to accelerate this process via parallelization.  In this section we show how to do so in the standard \emph{rank} and \emph{independence} oracle models for matroids.

\subsection{Matroid rank oracles} 
Given a rank oracle for the matroid, we get an algorithm that only makes $\O\left (\log(n)\log(k) \right )$ steps of matroid oracle queries and has polylogarithmic depth on a PRAM machine.  Recall that a rank oracle for $\M$ is given a set $S$ and returns its rank, i.e. the maximum size of an independent subset $T \subseteq S$.  The number of steps of matroid queries of an algorithm is the number of sequential steps it makes when polynomially-many queries to a matroid oracle for $\M$ can be executed in parallel in each step~\cite{karp1988complexity}.\footnote{More precisely, it allows $p$ queries per step and the results depend on $p$, we consider the case of $p = \poly(n)$.}  We use a parallel algorithm from~\cite{karp1988complexity} designed for constructing a base of a matroid with a rank oracle, and show  that it satisfies the random feasibility property.

\begin{algorithm}[H]
\caption{Parallel \blackbox \  for matroid constraint with rank oracle}
\begin{algorithmic}
 	\INPUT matroid $\M$, ground set $N$	
 	\STATE $b_1, \ldots, b_{|N|} \leftarrow $ random permutation of $N$
 	\STATE  $r_i \leftarrow \rank(\{b_1, \ldots,  b_{i}\})$, for all $i \in \{1, \ldots, n\}$
 	\STATE $a_i \leftarrow i$th $b_j$ s.t. $r_j - r_{j-1} = 1$
	\RETURN $a_1, \ldots, a_\ell$
  \end{algorithmic}
  \label{alg:rank}
\end{algorithm}

With Algorithm~\ref{alg:rank} as the \blackbox \ subroutine for \algone, we obtain the following result for matroid rank oracles (proof in Appendix~\ref{sec:appmatroid}).
\begin{restatable}{rThm}{thmrank}
For any $\epsilon>0$, there is an algorithm that obtains, with probability $1 - o(1)$,  a $1/2 - \O(\epsilon)$ approximation  with $\O\left(\log(n)\log\left(\frac{k}{\epsilon}\right)\frac{1}{\epsilon^2}\right)$ adaptivity and  steps of matroid rank queries.
\end{restatable}

This gives  $\O(\log(n) \log(k))$ adaptivity and steps of independence queries with $1-1/e-\epsilon$ approximation for maximizing the multilinear relaxation and $1/2-\epsilon$ approximation for maximizing a monotone submodular function under a matroid constraint.  In particular, we get polylogarithmic depth on a PRAM machine with a rank oracle.

%we constructed an algorithm with $O(\log(n) \log(k))$ adaptivity which obtains an  approximation arbitrarily close to $1-1/e$. Although the algorithm has low adaptivity, it is highly sequential with respect to the matroid constraint: \blackbox \ is independent of $f$ but has $O(k)$ iterations of matroid feasibility queries. 

%In this section, we study the number of rounds of matroid queries needed. 

\subsection{Matroid independence oracles}

Recall that an independence oracle for $\M$ is an oracle which given $S\subseteq N$ answers whether $S \in \M$ or $S \not \in M$.  We give a subroutine that requires $\tilde{\O}(n^{1/2})$ steps of independence matroid oracle queries and show that $\tilde{\Omega}(n^{1/3})$ steps are necessary.  Similar to the case of rank oracles we use a parallel algorithm from~\cite{karp1988complexity} for constructing a base of a matroid that can be used as the \blackbox \ subroutine while satisfying the random feasibility condition.   

%We get $\O(\log(n)\log(k))$ adaptivity with $\tilde{O}(n^{1/2})$ steps of matroid independence queries, with approximation of $1/2$ using \algone \ and $1-1/e$ with \algcontinuous .  Thus, while the $\Omega(n^{1/3})$ lower bound implies that it is impossible to obtain an algorithm with polylogarithmic steps of queries with an independence oracle, this 

\paragraph{$\tilde{O}(\sqrt{n})$ upper bound.}  We use the algorithm from \cite{karp1988complexity} for constructing a base of a matroid. %The important feature needed from the \blackbox \ subroutine is that it constructs a feasible sequence of elements while satisfying the random feasibility property. We designed \algone \ in the previous section in a manner so that the properties needed from \blackbox \ could be satisfied by the parallel matroid algorithms  for constructing a feasible sequence.

\begin{algorithm}[H]
\caption{Parallel \blackbox \  for matroid constraint with independence oracle}
\begin{algorithmic}
 	\INPUT matroid $\M$, ground set $N$	
 	\STATE $c \leftarrow 0, X \leftarrow N$
	\STATE \textbf{while} $|N| > 0$ \textbf{do}
	\STATE \ \ \ $b_1, \ldots, b_{|X|} \leftarrow $ random permutation of $X$
	\STATE \ \ \ $i^{\star} \leftarrow \max\{i : \{a_1, \ldots, a_c\} \cup \{b_1, \ldots, b_i\} \in \M\}$
	\STATE \ \ \ $a_{c + 1}, \ldots, a_{c + i^{\star}} \leftarrow b_1, \ldots, b_{i^{\star}}$	
	\STATE \ \ \ $c \leftarrow c + i^{\star}$
	\STATE \ \ \ $X \leftarrow \{a \in X :  \{a_1, \ldots, a_c, a\} \in \M\}$	
	\RETURN $a_1, \ldots, a_c$
  \end{algorithmic}
  \label{alg:3}
\end{algorithm}

With Algorithm~\ref{alg:3} as the \blackbox \ subroutine for \algone, we obtain the following result with independence oracles.  We defer the proof to Appendix~\ref{sec:appmatroid}.

\begin{restatable}{rThm}{thmindependence}
\label{thm:independence}
There is an algorithm that obtains, w.p. $1 - o(1)$,  a $1/2 - \O(\epsilon)$ approximation  with $\O\left(\log(n)\log\left(\frac{k}{\epsilon}\right)\frac{1}{\epsilon^2}\right)$ adaptivity and $O\left(\sqrt{n}\log(n)\log\left(\frac{k}{\epsilon}\right)\frac{1}{\epsilon^2}\right)$ steps of independence queries.
\end{restatable}

This gives  $\O(\log(n) \log(k))$ adaptivity and $\sqrt{n}\log(n)\log(k)$ steps of independence queries with $1-1/e-\epsilon$ approximation for maximizing the multilinear relaxation and $1/2-\epsilon$ approximation for maximizing a monotone submodular function under a matroid constraint.  In particular, even with independence oracles we get a sublinear algorithm in the PRAM model.  

\paragraph{$\tilde{\Omega}(n^{1/3})$ lower bound.} We show that there is no algorithm which obtains a constant approximation with less than $\tilde{\Omega}(n^{1/3})$ steps of independence queries, even for  a cardinality function $f(S) = |S|$. We do so by using the same construction for a hard matroid instance as in \cite{karp1988complexity} used to show an $\tilde{\Omega}(n^{1/3})$ lower bound on the number of steps of independence queries for constructing a base of a matroid.  Although the matroid instance is the same, we use a different approach since the proof technique of \cite{karp1988complexity} does not hold in our case (see proof and discussion in Appendix~\ref{sec:appmatroid}).

\begin{restatable}{rThm}{thmlowerbound}
\label{thm:lowerbound}
For any constant $\alpha$,  there is no algorithm with $\frac{n^{1/3}}{4 \alpha \log^2 n} - 1$ steps of $\poly(n)$ matroid queries which, w.p. strictly greater than $n^{-\Omega(\log n)}$,  obtains an $\alpha$ approximation for maximizing a cardinality function under a partition matroid constraint when given an independence oracle.
\end{restatable}

To the best of our knowledge, the gap between the  lower and upper bounds of $\tilde{Omega}(n^{1/3})$ and  $O(n^{1/2})$ parallel steps for  constructing a matroid basis given an independence oracle remains open since~\cite{karp1988complexity}.  Closing this gap for submodular maximization under a matroid constraint given an independence oracle is an interesting open problem that would also close the gap of~\cite{karp1988complexity}.

\newpage
\bibliographystyle{alpha}
 \bibliography{biblio}
 
 \newpage
 
\appendix

\section*{Appendix}

% !TeX root = main.tex

\section{Discussion about Additional Results}
\label{sec:explicit}
We discuss several cases for which our results and techniques generalize.

\paragraph{Cardinality constraint.} We first mention a generalization of \algone \ that is a $\O\left(\log(n)\right)$ adaptive algorithm that obtains a $1- 1/e - \O(\epsilon)$ approximation with probability $1 - o(1)$ for monotone submodular maximization under a cardinality constraint, which is the special case of a uniform matroid. Instead of sampling uniformly random subsets of $X$ of size $k/r$ as done in every iteration of the algorithm in \cite{BRS19}, it is possible to generate a single sequence and then add elements to $S$ and discard elements from $X$ in the same manner as \algone. We note that generating a random feasible sequence in parallel is trivial for a cardinality constraint $k$, one can simply pick $k$ elements uniformly at random. Similarly, the elements we add to the solution are approximately locally optimal and we discard a constant fraction of elements at every round. A main difference is that for the case of a cardinality constraint,  setting the threshold $t$ to  $t = (\OPT - f(S))/k$  is sufficient and, as shown in \cite{BRS19}, this threshold only needs a constant number of updates.  Thus, for the case of a cardinality constraint, we obtain a $\O(\log n)$ adaptive algorithm with a variant of \algone. In addition, the continuous greedy algorithm is not needed for a cardinality  constraint since adding elements with marginal contribution which approximates  $(\OPT - f(S))/k$ at every iteration guarantees a $1-1/e-\epsilon$ approximation.

\paragraph{Non-monotone functions.} For the case of maximizing a non-monotone submodular function under a cardinality constraint, similarly as for the monotone algorithm discussed above, we can also generate a single sequence instead of multiple random blocks of elements, as done in \cite{BBS18}.

\paragraph{Partition matroids with explicit representation.} Special families of matroids, such as graphical and partition matroids, have explicit representations. We consider the case where a partition matroid is given as input to the algorithm not as an oracle but with its explicit representation, meaning the algorithm is given the parts $P_1, \ldots, P_m$ of the partition matroid and the number  $p_1, \ldots, p_m$ of elements of each parts allowed by the matroid. 

For the more general setting  of packing constraints given to the algorithm as a collection of $m$ linear constraints, as previously mentioned, \cite{chekuri2018submodular} develop a $\O(\log^2(m) \log(n))$ adaptive algorithm that obtains with high probability a $1-1/e-\epsilon$ approximation, and has  polylogarithmic depth on a PRAM machine for partition matroids.

In this case of partition matroids, we obtain a $\O\left(\log(n)  \log\left(\frac{k}{\epsilon \lambda}\right)   \frac{1}{\epsilon \lambda}\right)$ adaptive algorithm that, with probability $1 - \delta$, obtains a $1-1/e - \O(\epsilon)$ approximation with $\lambda = \O\left(\epsilon^2 \log^{-1}\left(\frac{1}{\delta}\right)\right)$. This algorithm also has polylogarithmic depth. This algorithm uses \algcontinuous \ with the \blackbox \ subroutine for rank oracles since a rank oracle for partition matroids can easily be constructed in polylogarithmic depth when given the explicit representation of the matroid. As mentioned in \cite{chekuri2018submodular}, it is also possible to obtain a rounding scheme for partition matroids in polylogarithmic depth.

\paragraph{Intersection of $P$ matroids.} We formally analyze the more general constraint consisting of the intersection of $P$ matroids in Appendix~\ref{sec:appcomb}.

\section{Missing Proofs from Section~\ref{sec:comb}}
\label{sec:appcomb}

\subsection{Quasi-linear query complexity}

The query complexity of \algone \ and \algcontinuous \ can be improved from  $\O(nk \log(n) \log(k))$ to  quasi-linear with $\O(n \log(n) \log^2(k))$ queries if we allow $\O(\log(n) \log^2(k))$ rounds. This is done by finding $i^{\star}$ at every iteration of \algone \ by doing binary search of $i \in [\rank(\M(S,X))]$ instead of computing $X_i$ for all $i$ in parallel. Since there are at most $k$ values of $i$, this decrease the query complexity of finding  $i^{\star}$ from $nk$ to $n \log k$, but increases the adaptivity by $\log k$.

An important property to be able to perform binary search is to have $|X_i|$ decreasing in $i$. We show this with the following lemma.

\begin{lemma}
\label{lem:size}
At every iteration of \algone,  $X_{i+1} \subseteq X_i$ for all $i < \rank(\M(S, X))$.
\end{lemma}
\begin{proof}
Assume $a \in X_{i+1}$. Thus, $S \cup \{a_1, \ldots, a_{i}\} + a \in \M$ and $f_{S \cup \{a_1, \ldots, a_{i}\}}(a) \geq  t$. By the downward closed property of matroids, $S \cup \{a_1, \ldots, a_{i-1}\} + a \in \M$. By submodularity, $f_{S \cup \{a_1, \ldots, a_{i-1}\}}(a) \geq  f_{S \cup \{a_1, \ldots, a_{i}\}}(a) \geq  t$. We get that $a \in X_i$.
\end{proof}

\begin{corollary}
If \algone \ finds $i^{\star}$ by doing binary search, then its query complexity is $\O(n \log(n) \log^2(k))$and its adaptivity is  $\O(\log(n) \log^2(k))$.
\end{corollary}

\subsection{From expectation to high probability for the combinatorial algorithm}

We generalize  \algone \ to obtain an algorithm called \algoneplus, described below, which achieves a $1/2 - \epsilon$ approximation with high probability, instead of in expectation. We note that this generalization is not needed when \algone \ is used as a subroutine of \algcontinuous \ for the $1-1/e-\epsilon$ result.

\begin{algorithm}[H]
\caption{\algoneplus,  \algone  \ with high probability guarantee}
\begin{algorithmic}
    	\INPUT function $f$, feasibility constraint $\M$
    	\STATE  $S \leftarrow \emptyset, t \leftarrow \max_{a \in N}f(a)$
    	\STATE \textbf{for} $\Delta$ iterations \textbf{do}
	\STATE \ \  \ $X \leftarrow N$
	\STATE \ \ \  \textbf{while} $X \neq \emptyset$ \textbf{do}
	\STATE \ \ \ \ \ \ \textbf{for} $j = 1$ to $\rho$  \textbf{do (non-adaptivity and in parallel)}
	\STATE \ \ \ \ \ \  \ \ \ $a_1, \ldots, a_{\rank(\M(S, X))} \leftarrow \blackbox(\M(S,X))$
	\STATE \ \ \ \ \ \  \ \ \ $X_i \leftarrow \{a \in X :  S \cup \{a_1, \ldots, a_{i}, a\} \in \M \text{ and } f_{S \cup \{a_1, \ldots, a_{i}\}}(a) \geq  t\}$
	\STATE \ \ \  \ \ \  \ \ \ $i^{\star} \leftarrow \min\left\{i : |X_i| \leq (1 - \epsilon)|X|\right\}$ 
	\STATE \ \ \ \ \ \   \ \ \ $S^j \leftarrow S \cup \{a_1, \ldots, a_{i^{\star}}\}$
	\STATE \ \ \ \ \ \  \ \ \ $X^j \leftarrow  X_{i^{\star}}$
	\STATE \ \ \ \ \ \  \ \ \ $v^j \leftarrow \frac{1}{i^{\star}} \sum_{\ell = 1}^{i^{\star}} f_{S \cup \{a_1, \ldots, a_{\ell - 1}\}}(a_\ell)$
	\STATE \ \ \ \ \ \ $j^{\star} \leftarrow \argmax_{j \in [\rho]} v^j$
	\STATE  \ \ \ \ \ \ $S \leftarrow S^j$
	\STATE  \ \ \ \ \ \ $X \leftarrow X^j$
	\STATE \ \ \ $t \leftarrow (1 - \epsilon) t$
    	\RETURN $S$ 
  \end{algorithmic}
  \label{alg:1/2hp}
\end{algorithm}

\thmcomb*
\begin{proof}
We set $\Delta = \O\left(\frac{1}{\epsilon} \log\left(\frac{k}{\epsilon}\right) \right)$. Initially we have $t_i \leq \OPT$. After $\Delta$ iterations of \algone, the final value of $t$ is $t_f \leq (1 - \epsilon)^{\Delta}\OPT = \O\left(\frac{\epsilon}{k}\right) \OPT$. We begin by adding dummy elements to $S$ so that $|S| = k$, which enables pairwise comparisons between $S$ and $O$. In particular, we consider $S'$, which is $S$ together with $\rank(\M) - |S|$ dummy elements $a_{|S| + 1}, \ldots a_k$ such that, for any $T$, $f_T(a) = t_f$. Thus, by Lemma~\ref{lem:invariants}, for dummy elements $a_i$, $f_{S_{i-1}}(a_i) = t_f \geq (1 - \epsilon) \max_{a : S_{i-1} \cup a \in \M} f_{S_{i-1}}(a)$.

By Lemma~\ref{lem:adaptivity}, there are $\O(\Delta \log(n) / \epsilon)$ iterations of the while-loop. Since each iteration of the while-loop is non-adaptive, \algoneplus \ is $\O(\Delta \log(n) / \epsilon)$ adaptive

Consider an iteration of the while-loop of \algoneplus.  We first argue that for each inner-iteration $j$, $\sum_{i \in [i^{\star}]}f_{S_{i-1}}(a_i) \geq (1 - \epsilon)^2 i^{\star} t$. 
We first note that $\Pr_{a_i}\left[f_{S \cup \{a_1, \ldots, a_{i-1}\}}(a_i) \geq t\right] \geq 1 - \epsilon$ by the definition of $i^{\star}$ and the random feasible sequence property. Let $Y$ be the number of indices $i \leq i^{\star}$ such that $f_{S \cup \{a_1, \ldots, a_{i-1}\}}(a_i) \geq t$. By Chernoff bound, with $\mu = \E[Y] \geq  (1 - \epsilon)i^{\star}$
$$\Pr\left[Y \leq (1 - \epsilon)(1 - \epsilon)i^{\star}\right] \leq e^{-\epsilon^2 (1 - \epsilon) i^{\star}/2} \leq e^{-\epsilon^2 (1 - \epsilon)/2}.$$

Let $Z \leq \rho$ be the number of inner-iterations $j$ such that $Y \geq (1 - \epsilon)(1 - \epsilon)i^{\star}$. By Chernoff bound, with $\mu = \E[Z] \geq (1 - e^{-\epsilon^2 (1 - \epsilon)/2})\rho$, 
$$\Pr\left[Z \leq \frac{1}{2}(1 - e^{-\epsilon^2 (1 - \epsilon)/2})\rho\right] \leq e^{(1 - e^{-\epsilon^2 (1 - \epsilon)/2})\rho/8}.$$

Thus, with $\rho = \O\left(\frac{1}{1 - e^{-\epsilon^2 }} \log\left(\frac{\Delta \log n}{\epsilon \delta}\right)\right)$, we have that with probability $1 - \O\left(\epsilon \delta/(\Delta \log n)\right)$, there is at least one inner-iteration $j$ such that $Y \geq (1 - \epsilon)(1 - \epsilon)i^{\star}$. Thus $\sum_{i \in [i^{\star}]}f_{S_{i-1}}(a_i) \geq (1 - \epsilon)^2 i^{\star} t$. By Lemma~\ref{lem:invariants},
$$\sum_{i \in [i^{\star}]}f_{S_{i-1}}(a_i) \geq (1 - \epsilon)^3  \max_{a : S_{i-1} \cup a \in \M} f_{S_{i-1}}(a).$$
By a union bound, this holds over all iterations
of the while-loop of \algoneplus \ with probability $ 1- \delta$ and we get that 
$$\sum_{i \in [k]}f_{S_{i-1}}(a_i) \geq (1 - \epsilon)^3  \max_{a : S_{i-1} \cup a \in \M} f_{S_{i-1}}(a).$$ Let $O = \{o_1, \ldots, o_k\}$ such that $\{a_1, \ldots, a_{i-1}, o_i\}$ is feasible for all $i$, which exists by the augmentation property of matroids. We conclude that with probability $ 1- \delta$, 
\begin{align*}
f(S') & = \sum_{i \in [k]}\E[f_{S_{i-1}}(a_i)] \\
& \geq (1 - \epsilon)^3  \max_{a : S_{i-1} \cup a \in \M} f_{S_{i-1}}(a) \\
& \geq (1-\epsilon)^3\sum_{i \in [k]}\E[f_{S_{i-1}}(o_i)] \\
& \geq (1-\epsilon)^3f_{S'}(O) \\
& \geq (1-\epsilon)^3(\OPT - f(S'))
\end{align*}
and since $$f(S) = f(S') - (\rank(\M) -|S|) t_f \geq f(S') - \O(\epsilon) \OPT,$$ we conclude that $f(S) \geq (1/2 - \O(\epsilon)) \OPT$.
\end{proof}

\subsection{Intersection of matroid constraints}

We consider constraint $\M = \cap_{i=1}^P \M_i$ which is the intersection of $P$ matroids $\M_i$, i.e. $S \in \M$ if $S \in \M_i$ for all $i \leq P$. Similarly as for a single matroid constraint, we denote the size of the largest feasible set by $k$.  We denote the rank of a set $S$ with respect to matroid $\M_j$ by $\rank_j(S)$. We define $\Span_j(S)$, called the \emph{span} of $S$ in $\mathcal{M}_j$ by:
\[\Span_j(S) = \{a \in N: \rank_j(S \cup a) = \rank_j(S) \}\]

We will use the following claim.
\begin{claim}[Prop. 2.2 in \cite{NWF78}]\label{clm:fisher}
    If for $\forall t\in[k]$ $\sum_{i = 0}^{t - 1} \sigma_i \le t$  and $p_{i - 1} \ge p_i$, with $\sigma_i, p_i \ge 0$ then:
    \[\sum_{i = 0}^{k - 1} p_i \sigma_i \le  \sum_{i = 0}^{k - 1} p_i.\]
  \end{claim}

Similarly as for a single matroid, we give the approximation guaranteed obtained by a solution $S$ with near-optimal marginal contributions for each $a \in S$. 

\begin{lemma}
\label{lem:matroidintersection}
Assume that $S = \{a_1, \ldots, a_k\}$ such that 
$$f_{S_{i-1}}(a_i) \geq (1 - \epsilon) \max_{a: S_{i-1} \cup a \in \M}f_{S_{i-1}}(a)$$ where $S_i = \{a_1, \ldots, a_i\}$. Then, if $\M$ is the intersection of $P$ matroids, we have $$f(S) \geq \left(\frac{1}{P+1} - \O(\epsilon)\right)\OPT.$$
\end{lemma}
\begin{proof}
Since $S_i$ and $O$ are independent sets in $\mathcal{M}_{j}$ we have:
\[\texttt{rank}_j ( \Span_j(S_i) \cap O ) = |\Span_j(S_i) \cap O | \leq |\Span_j(S_i)| = |S_i|\leq i \]

Define $U_i = \cup_{j=1}^{P} \Span_j(S_i)$, to be the set of elements which are not part of the maximization at index $i+1$ of the procedure, and hence cannot give value at that stage. We have:
\[|U_i \cap O| = |(\cup_{j = 1}^{P} \Span_j(S_i)) \cap O| \le \sum_{j=1}^{P} |\Span_j(S_i) \cap O| \le P\cdot  i\]
Let $V_i = (U_i \setminus U_{i-1}) \cap O$ be the elements of $O$ which are not part of the maximization at index $i$, but were part of the maximization at index $i-1$.  If $a \in V_{i}$ then it must be that
$$(1 - \epsilon)f_{S_k}(a) (1 - \epsilon)\leq f_{S_{i-1}}(a) \le \max_{b : S_{i-1} \cup b \in \M}f_{S_{i-1}}(b) $$
where the first inequality is due to submodularity of $f$. Hence, we can upper bound:
\begin{align*}
\sum_{o \in O \setminus S_{k}} f_{S_{k}}(o) \le \sum_{i = 1}^{k} \sum_{o \in V_i}\max_{a : S_{i-1} \cup a \in \M}f_{S_{i-1}}(a)
= \sum_{i = 1}^{k} |V_i| \max_{a : S_{i-1} \cup a \in \M}f_{S_{i-1}}(a) \le P \sum_{i = 1}^{k} \max_{a : S_{i-1} \cup a \in \M}f_{S_{i-1}}(a)
 \end{align*}
where the last inequality uses $\sum_{t=1}^{i}|V_t| = |U_i \cap O| \le Pi$ and the  claim due to~\ref{clm:fisher}.
Together with $\OPT \le f(O \cup S_{k}) \le f(S_{k}) + \sum_{o \in O \setminus S_{k}} f_{S_{k}}(o)$ and $f_{S_{i-1}}(a_i) \geq (1- \epsilon) \max_{a : S_{i-1} \cup a \in \M}f_{S_{i-1}}(a)$ we get:
$$f(S) \geq \left(\frac{1}{P+1} - \O(\epsilon)\right)\OPT.$$
as required.
\end{proof}

Since Lemma~\ref{lem:invariants}  only uses the downward closed property of $\M$ and since intersections of matroids are downward closed, \algoneplus \ obtains a solution $S$ with near-optimal marginal contributions for each $a_i \in S = \{a_1, \ldots, a_k\}$. Combined with the previous lemma, we obtain the result for intersections of matroids.
\begin{theorem}
\label{thm:intersection}
For any $\epsilon>0$, \algoneplus \ is an $\O\left(\log(n)\log\left(\frac{k}{\epsilon}\right)\frac{1}{\epsilon^2}\right)$ adaptive algorithm that obtains a $1/(P+1) - \O(\epsilon)$ approximation with probability $1 - o(1)$ for maximizing a monotone submodular function under the intersection of $P$ matroids.
\end{theorem}
\begin{proof}
The first part of the of the proof follows similarly as the proof for Theorem~\ref{thm:hp} by using  Lemma~\ref{lem:invariants}, which also hold for intersections of matroids,  to obtain the near-optimal marginal contributions of each $a_i \in S$ with probability $1 - o(1)$:

$$\sum_{i \in [i^{\star}]}f_{S_{i-1}}(a_i) \geq (1 - \epsilon)^3  \max_{a : S_{i-1} \cup a \in \M} f_{S_{i-1}}(a).$$
We then combine this with Lemma~\ref{lem:matroidintersection} to obtain the $1/(P+1) - \O(\epsilon)$ approximation with probability $1 - o(1).$
\end{proof}
\section{Missing Proofs from Section~\ref{sec:continuous}}
\label{sec:applayerone}

\paragraph{Discussion on constant step size $\lambda$.} In contrast to the continuous greedy, the accelerated continuous greedy uses \emph{constant} steps sizes $\lambda$ to guarantee low adaptivity.  The challenge with using constant $\lambda$ is that $F_{\bx}(\lambda S)$ is arbitrarily low with  $S := \argmax_{T \in \M} \sum_{a \in T} g(a)$ due to the overlap in value of elements $a$ with high individual value $g(a)$.

For example, consider ground set $N = A \cup B$ with $$f(S) = \min(\log n, |S \cap A|) + |S \cap B|,$$
 $\bx = \zeros$ and $S = A$. With $\lambda = 1/n$, we note that sampling $R \sim \lambda A$ where $R$ independently contains each element in $S$ with probability $1/n$  gives $|R| \leq \log n$ with high probability and  we get $F_{\bx}(\lambda A) = (1-o(1))|A|$, which is near-optimal for a set of size $|A|$. However, with constant $\lambda$, then sampling $R \sim \lambda A$
gives $|R| > \log n$ with high probability. Thus  $F_{\bx}(\lambda A) \leq \log(n)$ which is arbitrarily far from optimal for $|A| = |B| >> \log n$ since $F_{\bx}(\lambda B) = \lambda |B|$.

\lemmetaoneExp*
\begin{proof}
First, $\bx \in P$ since it is a convex combinations of $\lambda^{-1}$ vectors $\ones S$ with $S \in \M$. Next, let $\bx_i$ denote the solution $\bx$ at the $i$th iteration of \algcontinuous.
The algorithm increases the value of the solution $\bx$ by at least $(1- \epsilon) \cdot \lambda \cdot  \left( \OPT - F(\bx) \right) $ at every iteration. Thus,
$$F(\bx_i) \geq F(\bx_{i-1}) + (1- \epsilon)\cdot \lambda \cdot \left(\OPT - F(\bx_{i-1})\right).$$
Next, we show by induction on $i$ that
$$F(\bx_{i}) \geq \left( 1 - \left(1 - (1-\epsilon) \lambda \right)^i\right) \OPT.$$
Observe that
\begin{align*}
F(\bx_{i}) & \geq F(\bx_{i-1}) + (1-\epsilon) \lambda \left(\OPT - F(\bx_{i-1})\right) \\
& = (1-\epsilon) \lambda \OPT + \left(1- (1-\epsilon) \lambda \right) F(\bx_{i-1})\\
& \geq (1-\epsilon) \lambda  \OPT + \left(1- (1-\epsilon) \lambda \right)\left( 1 - \left(1 - (1-\epsilon) \lambda \right)^{i-1}\right) \OPT\\
& = \left( 1 - \left(1 - (1-\epsilon) \lambda \right)^i\right) \OPT
\end{align*}
Thus, with $i = \lambda^{-1}$, we return solution $\bx = \bx_{\lambda^{-1}}$ such that
$$F(\bx) \geq \left( 1 - \left(1 - (1-\epsilon) \lambda \right)^{\lambda^{-1}}\right) \OPT.$$
Next, since $ 1- x \leq e^{-x}$ for all $x \in \R$,
$\left(1 - (1-\epsilon) \lambda \right)^{\lambda^{-1}} \leq \left(e^{-(1-\epsilon)\lambda}\right)^{\lambda^{-1}} = e^{-(1-\epsilon)}.$ 
  We conclude that  
$$F(\bx) \geq \left(1 - e^{-(1-  \epsilon)}\right)\OPT = \left(1 - \frac{e^{ \epsilon}}{e}\right)\OPT \geq \left(1 - \frac{1 + 2 \epsilon}{e}\right)\OPT  \geq \left(1 - \frac{1}{e} - \epsilon \right)\OPT$$ where the second inequality is since $e^{x} \leq 1 + 2 x$ for $0 < x < 1$. 
\end{proof}

\lemmatroidO*
\begin{proof}
The proof is by reverse induction. For $i = k$, we have $S_i \cup O_{i+1:k} = S_k = S \in \M$ by Lemma~\ref{lem:invariants}. Consider $i < k$ and assume that $S_{i+1} \cup O_{i+2:k} \in \M$ for some ordering $o_{i+2}, \ldots, o_k$ of $O_{i+2:k}$. By the downward closed property of matroids, $S_{i} \cup O_{i+2:k} \in \M$. By the augmentation property of matroids, there exists $o_{i+1} \in O \setminus (S_{i} \cup O_{i+2:k})$ such that $S_{i} \cup O_{i+2:k} + o_{i+1} = S_{i} \cup O_{i+1:k} \in \M$. 
\end{proof}

\thmexpectation*
\begin{proof}
We use step size $\lambda = \O\left(\epsilon\right)$ for \algcontinuous \ and $\Delta = \O\left(\frac{1}{\epsilon}\log\left(\frac{k}{\epsilon \lambda}\right)\right)$ outer-iterations for \algone. Thus, by Lemma~\ref{lem:adaptivity}, the adaptivity is $\O\left(\frac{\Delta \log n}{\lambda\epsilon}\right) =  \O\left(\log(n)  \log\left(\frac{k}{\epsilon^2}\right)   \frac{1}{\epsilon^2}\right).$ By Lemma~\ref{lem:maincontinuous}, we have $\E[F_{\bx}(\delta S)] \geq (1 - \O(\epsilon))\lambda(\OPT - F(\bx))$ at every iteration $i$.  Combining with Lemma~\ref{lem:metaoneExp}, we obtain that  $\E[F(\bx)] \geq (1-e^{-1} - \O(\epsilon))\OPT$. 

It remains to round the solution $\bx$. We note that  there exist rounding schemes with arbitrarily small loss that are independent of the function $f$  \cite{chekuri2009dependent,vondrak2011submodular} (so they do not perform any queries to $f$). The set $S$ we obtain from rounding the solution $\bx$ returned by \algcontinuous \ with these techniques is thus a $1-1/e - \O(\epsilon)$ approximation with no additional adaptivity.
\end{proof}

\lemmetaone*
\begin{proof}
First, $\bx \in P$ since it is a convex combinations of $\lambda^{-1}$ vectors $\ones_S \in \M$. Next, let $\bx_i$ denote the solution $\bx$ at the $i$th iteration of \algcontinuous.
The algorithm increases the value of the solution $\bx$ by at least $\alpha_i \cdot \lambda \cdot  \left( \OPT - F(\bx) \right) $ at every iteration. Thus,
$$F(\bx_i) \geq F(\bx_{i-1}) + \alpha_i \cdot \lambda \cdot \left(\OPT - F(\bx_{i-1})\right).$$
Next, we show by induction on $i$ that
$$F(\bx_{i}) \geq \left( 1 - \prod_{j=1}^i\left(1 - \lambda \alpha_j \right)\right) \OPT.$$
Observe that
\begin{align*}
F(\bx_{i}) & \geq F(\bx_{i-1}) + \alpha_i \lambda \left(\OPT - F(\bx_{i-1})\right) \\
& = \alpha_i \lambda \OPT + \left(1- \alpha_i \lambda \right) F(\bx_{i-1})\\
& \geq \alpha_i \lambda \OPT + \left(1- \alpha_i \lambda \right)\left( 1 - \prod_{j=1}^{i-1}\left(1 - \lambda \alpha_j \right)\right) \OPT\\
& = \alpha_i \lambda \OPT + \left(1- \alpha_i \lambda  - \prod_{j=1}^{i}\left(1 - \lambda \alpha_j \right)\right) \OPT\\
& = \left(1  - \prod_{j=1}^{i}\left(1 - \lambda \alpha_j \right)\right) \OPT
\end{align*}
where the first inequality is by the assumption of the lemma, the second by the inductive hypothesis, and the equalities by rearranging the terms. 
Thus, with $i = \lambda^{-1}$, we return solution $\bx = \bx_{\lambda^{-1}}$ such that
$$F(\bx) \geq \left(1  - \prod_{j=1}^{\lambda^{-1}}\left(1 - \lambda \alpha_j \right)\right) \OPT.$$
Since $ 1- x \leq e^{-x}$ for all $x \in \R$,
\begin{align*}
1  - \prod_{j=1}^{\lambda^{-1}}\left(1 - \lambda \alpha_j \right)  \geq 1  - \prod_{j=1}^{\lambda^{-1}}e^{ - \lambda \alpha_j} 
 = 1  - e^{ - \lambda \sum_{j=1}^{\lambda^{-1}} \alpha_j } 
\geq 1  - e^{ - (1 - \epsilon)} \geq 1 - e^{-1} - 2\epsilon/e \geq 1- e^{-1} - \epsilon
\end{align*}
where the second inequality is since $e^{x} \leq 1 + 2x$ for $0 < x < 1$.
\end{proof}

\thmhp*
\begin{proof}
We use $\Delta = \O\left(\frac{1}{\epsilon}\log\left(\frac{k}{\epsilon \lambda}\right)\right)$ outer-iterations for \algone. Thus, by Lemma~\ref{lem:adaptivity}, the adaptivity is $\O\left(\frac{\Delta \log n}{\lambda\epsilon}\right) =  \O\left(\log(n)  \log\left(\frac{k}{\epsilon \lambda}\right)   \frac{1}{\epsilon \lambda}\right).$

By Lemma~\ref{lem:maincontinuous}, we have $F_{\bx}(\delta S) \geq \alpha_i\lambda(\OPT - F(\bx))$ at every iteration $i$ with $\E\left[\alpha_i\right] \geq 1 - \epsilon'$ where $\epsilon' = \O(\epsilon)$. By a Chernoff bound with $\E[\lambda \sum_{i \in \lambda^{-1}} \alpha_i] \geq 1-\epsilon'$, 
$$\Pr\left[\lambda \sum_{i \in [\lambda^{-1}]} \alpha_i < (1 - \epsilon)  (1-\epsilon')\right] \leq e^{-\epsilon^2(1-\epsilon') \lambda^{-1}/2}.$$
Thus, with probability $p = 1 - e^{-\epsilon^2(1-\epsilon') \lambda^{-1}/2}$, $\lambda \sum_{i \in [\lambda^{-1}]} \alpha_i \geq 1 - \epsilon - \epsilon'$.  By Lemma~\ref{lem:metaone}, we conclude that w.p. $p$, $F(\bx) \geq (1-e^{-1} - (\epsilon + \epsilon'))\OPT$. With step size $\lambda = O(\epsilon^2 / \log(1/\delta))$, we get that with probability $1 - \delta$, $F(\bx) \geq (1- e^{-1} - O(\epsilon))\OPT$. 

It remains to round the solution $\bx$. We note that  there exist rounding schemes with arbitrarily small loss that are independent of the function $f$  \cite{chekuri2009dependent,vondrak2011submodular} (so they do not perform any queries to $f$). The set $S$ we obtain from rounding the solution $\bx$ returned by \algcontinuous \ with these techniques is thus a $1-1/e - \O(\epsilon)$ approximation with no additional adaptivity.
\end{proof}
\section{Missing Analysis from Section~\ref{sec:matroid}}
\label{sec:appmatroid}

\subsection{Lower bound on steps of independence queries}
We first give the construction from  \cite{karp1988complexity}. The partition matroid has $p = n^{1/3}/ \log^2 n$ parts $P_1, \ldots, P_{p}$ of equal size $n^{2/3} \log^2 n$ and a set $S$ is independent if $|S \cap P_i| \leq i n^{1/3}  \log^2 n$ for all parts $P_i$. Informally, the hardness is since an algorithm cannot learn part $P_{i+1}$ in $i$ steps of independence queries. 

We  lower bound the  performance of any algorithm against a matroid chosen uniformly at random over all such partitions $P_1, \ldots, P_{p}$.

The issue with applying the approach in  \cite{karp1988complexity} is that when it considers a query $B$ at some step $j < i$, the analysis bounds the intersection of fixed query $B$ with uniformly random parts $P_i, \ldots, P_{p}$ of $N \setminus \cup_{j=1}^{i-1} P_j$. However, a query at step $j < i$ is not independent  of the randomization $P_i, \ldots, P_{p}$ over $N \setminus \cup_{j=1}^{i-1} P_j$. For example, consider a query $T$ at step $j-1$ such that its intersection with $N \setminus \cup_{j=1}^{i-2} P_j$ is of size  $in^{1/3} \log^2 n + 1$ and the oracle answers $S \in \M$. This implies that $T \not \subseteq P_i$ and thus for a fixed query $B$ at step $j$, $P_i$ is not a random part since it cannot be such that $T \not \subseteq P_i$. Instead, we use an approach which argues about independence of queries with some parts $P_i, \ldots, P_{p}$ under some conditioning on $P_i, \ldots, P_{p}$.

We introduce some notation. Let $\M(S)$ be the indicator variable for $S \in \M$. We denote by $S^j$ the elements in $S$ that are not in a part $P_i$ with $i < j$, i.e. $S^j := S \setminus \{P_1, \ldots, P_{j-1}\}$. We say that a set $S$, assuming $|S^i| \geq n^{1/3} \log^2 n$,  concentrates at step $i$ if 
$$|S^i \cap P_j| \leq  \frac{(1+1/8i)|S^i|}{p-i}$$
for all $j > i$ and we use $c(S, i)$ for the indicator variable for $S$ concentrating at step $i$. Finally, $I(S, i)$ indicates if, for some $P_1, \ldots, P_{i} $, the answer $\M(S)$ of the independence oracle $\M$ to query $S$ is independent of the randomization of parts $P_{i+1}, \ldots, P_{p}$ over $N \setminus \cup_{j=1}^iP_j$. The main lemma is the following.

\begin{lemma}
\label{lem:ind}
For any $i \in [p]$, with probability at least $1 - n^{-\Omega(\log n)}$ over $P_1, \ldots, P_{i}$, for all queries $S$ by an algorithm with $i$ steps of  queries, the answer $\M(S)$ of the oracle is independent of $P_{i+1}, \ldots, P_{p}$, conditioned on $S$ concentrating over these parts, i.e., for all queries $S$ at step $i$
$$\Pr_{P_1, \ldots, P_{i}}\left[I(S, i) | c(S, i)\right] \geq 1 - n^{-\Omega(\log n)}.$$
\end{lemma}
\begin{proof}
The proof is by induction on $i$. Consider $i > 0$ and assume that for all queries $S$ at step $j < i$, $\Pr_{P_1, \ldots, P_{j}}\left[I(S, j) | c(S, j)\right] \geq 1 - n^{-\Omega(\log n)}$. By a union bound, with probability $1 - n^{-\Omega(\log n)}$, this holds for all queries at all steps $j < i$ and we assume this is the case.

Consider some set $S$ and a part $P_j$ which is a uniformly random subset of $N \setminus \cup_{\ell=1}^{i-1}P_\ell$, with $j \geq i$. Then $\E_{P_j}\left[|S^i \cap P_j|\right] = |S^i| / (p-i)$ and by Chernoff bound we get that 
\begin{align*}
\Pr_{P_j}\left[|S^i \cap P_j| \leq (1+1/8i) |S^i| / (p-i)\right] & \geq 1 - e^{-(1/8i)^2  |S^i| / 2 (p-i)} \\
& \geq 1 - e^{-(1/8i)^2 n^{1/3} \log^2 n / 2 } \\
& =  1 - e^{ - \Omega(\log^2n)}\\
& \geq 1 - n^{-\Omega(\log n)},
\end{align*}
assuming that $|S^i| / (p-i) \geq n^{1/3} \log^2 n$ and since $i \leq n^{1/3}$. By a union bound, for all queries $S$ at some step $j < i$, $|S^i \cap P_\ell| \leq (1+1/8i) |S^i| / (p-i)$ for all $j < \ell < i$ with probability $1 - n^{-\Omega(\log n)}$ over the randomization of $P_1, \ldots, P_{i-1}$ and we assume this is the case.

The answer of query $S$ at step $j < i$ is independent of the randomization of parts $P_{j+1}, \ldots, P_{p}$ over $N \setminus \cup_{\ell=1}^jP_\ell$ conditioned on these parts concentrating. Since we assumed parts $P_{j+1}, \ldots, P_{i-1}$ concentrate with $S$, it is  independent of the randomization of  parts $P_{i}, \ldots, P_{p}$ over $N \setminus \cup_{\ell=1}^{i-1}P_\ell$ conditioned on these parts concentrating.

Thus, the decision of the algorithm to query a set $S$ at step $i$ is independent of  the randomization of $P_i, \ldots, P_{p}$, conditioned on these parts concentrating with previous queries. 

We first consider a uniformly random part $P_i$ over $N \setminus \cup_{\ell=1}^{i-1}P_\ell$. There are two cases for a query $S$ at step $i$
\begin{itemize}
\item If $|S^i| > (1 + 1/4i) i  n^{1/3}(p-i)$. Then  by Chernoff bound with $\mu = \E_{P_i}\left[|S^i \cap P_i|\right]  = (1 + 1/4i) i  n^{1/3}\log^2 n$, 
$$\Pr_{P_i}\left[|S^i \cap P_i| \leq i n^{1/3} \log^2 n\right] =  n^{-\Omega(\log n)}.$$ Thus, with probability $1 - n^{-\Omega(\log n)}$, $S \not \in \M$ and this is independent of the randomization of $P_{i+1}, \ldots, P_{p}$.
\item If $|S^i| \leq (1+1/4i) i n^{1/3}(p-i)$. Then, if $S$ concentrates with parts $P_{i+1}, \ldots, P_{p}$, by definition, $|S^i \cap P_j| \leq (1+1/8i) \frac{|S^i|}{p-i}$ and 
$$|S^i \cap P_j| \leq (1+1/4i)(1+1/8i)  i n^{1/3}\log^2 n < (i + 3/4) n^{1/3} \log^2 n < j  n^{1/3} \log^2 n$$
for $j > i$ and $S$ is feasible with respect to part $P_j$. So $\M(S)$ is independent of the randomization of $P_{i+1}, \ldots, P_{p}$, conditioned on $c(S,i)$.
\end{itemize}
The last piece needed from $P_i$ is that, due to the conditioning, it must concentrate with all queries from previous steps.
 As previously, this is the case with probability $1 - n^{-\Omega(\log n)}$. Combined with the above, we obtain
$\Pr_{P_1, \ldots, P_{i}}\left[I(S, i) | c(S, i)\right] \geq 1 - n^{-\Omega(\log n)}$.
\end{proof}

\thmlowerbound*
\begin{proof}
Consider a uniformly random partition of the ground set in parts $P_1, \ldots, P_{p}$ with $p = n^{1/3} / \log^2 n $ each of size $n^{2/3} \log^2 n$, the partition matroid overt these parts described previously, and the simple function $f(S) = |S|$. By a similar Chernoff bound as in Lemma~\ref{lem:ind}, we have that $\Pr_{P_j}\left[|S^i \cap P_j| \leq (1+1/8i) |S^i| / (p-i)\right] \geq 1 - n^{-\Omega(\log n)}$ for all queries at $S$ at step $i$ and $j > i$. Thus $\Pr[c(S,i)] \geq 1 - n^{-\Omega(\log n)}$ for query $S$ at step $i$ and by a union bound this holds for all queries by the algorithm. Thus, by Lemma~\ref{lem:ind}, we have that for all queries $S$ by an $p / (4 \alpha)- 1$  steps algorithm, the answer of the oracle to query $S$ is independent of the randomization of $P_{p / (4 \alpha)}, \ldots, P_{p}$ with probability $1 - n^{-\Omega(\log n)}$, conditioned on these parts concentrating with the queries, which they do with probability $1 - n^{-\Omega(\log n)}$. 

Thus, the solution $S$ returned by the algorithm after $p / (4 \alpha) - 1$  steps of matroid queries is independent of the randomization of $P_{p / (4 \alpha)}, \ldots, P_{p}$ with probability $1 - n^{-\Omega(\log n)}$, conditioned on these parts concentrating with the queries.

Assume the algorithm returns $S$ such that $|S^{p / (4 \alpha)}| > (1+1/8n^{1/3})(1-1/(4\alpha)) p n^{2/3} \log^2 n / (4\alpha).$ Thus, with $P_{p/ (4 \alpha)+1}$ a random part of $N \setminus \cup_{\ell=1}^{p / (4 \alpha) - 1}P_\ell$, $$\E_{P_{p / (4 \alpha)}}\left[|S \cap P_{p / (4 \alpha)+1}|\right] =  (1+1/8n^{1/3})  n^{2/3} \log^2 n/(4 \alpha).$$ By Chernoff bound, we have that with probability $1 - n^{-\Omega(\log n)}$, 
$$\Pr_{P_{p / (4 \alpha)}}\left[|S \cap P_{p / (4 \alpha)}| >   n^{2/3} \log^2 n /(4 \alpha)\right] \geq 1 - n^{-\Omega(\log n)}$$
and thus $S \not \in \M$.

If the algorithm returns $S$ such that $|S^{p / (4 \alpha)}| \leq (1+1/8n^{1/3})(1-1/(4\alpha)) p n^{2/3} \log^2 n / (4\alpha).$ Then, if $S \in \M$, there are at most $p/(4\alpha) \cdot n^{1/3} \log^2(n) p /(4\alpha)$ elements in $S$ from the first $p/(4\alpha)$ parts. Thus $$|S| \leq  (1+1/8i)(1-1/(4\alpha)) n / (4\alpha) + p/(4\alpha) \cdot n^{1/3} \log^2(n) p/(4\alpha) < n/(2 \log^2(n) \alpha).$$

Note that a base $B$ for the matroid has size 
$$|B| = \sum_{i=1}^{p} i n^{1/3}\log^2n = n^{1/3} n^{1/3}(n^{1/3}/\log^2 n +1)/2 > n / (2\log^2 n).$$

The parts $P_{n^{1/3} / (4 \alpha)}, \ldots, P_{p}$ concentrate with all the queries with probability $1 - n^{-\Omega(\log n)}$. Thus, the algorithm returns, with probability $1 - n^{-\Omega(\log n)}$, either a set $S \not \in \M$ or a set $S$ such that $|S| < |B| / \alpha$. Thus there is at least one instance of parts $P_1, \ldots, P_{p}$ such that the algorithm does not obtain an $\alpha$ approximation with probability strictly greater than $n^{-\Omega(\log n)}$.
\end{proof}

\subsection{An algorithm with  $\tilde{O}(\sqrt{n})$ steps of independence queries}

We show that  Algorithm~\ref{alg:3} satisfies the random feasibility condition required by \algone.

\begin{lemma}
\label{lem:feasibility1}
 Algorithm~\ref{alg:3} satisfies the random feasibility condition.
\end{lemma}
\begin{proof}
Consider $a_i$ for $i \leq c$. By the algorithm, we have $a_i = b_j$ for some $j \leq i^{\star}$ at some iteration $\ell$ with $b_1, \ldots, b_{|X_{\ell}|}$ and $c_{\ell}$. By the definition of $b_j$ and $X$, we have that $b_j$ is a uniformly random elements among all elements $X_\ell \setminus \{b_1, \ldots, b_{j-1}\}.$ Conditioned on $i^{\star} \geq j$, we have that $\{a_1, \ldots, a_{c_\ell}, b_1, \ldots, b_j\} \in \M$. By the downward closed property of matroids, $\{a \in X :   \{a_1, \ldots, a_{c_\ell}, b_1, \ldots, b_{j-1}, a\} \in \M\} \subseteq X_\ell \setminus \{b_1, \ldots, b_{j-1}\}$. Thus $b_j = a_i$ is uniformly random over all $a \in X$ such that $\{a_1, \ldots, a_{c_\ell}, b_1, \ldots, b_{j-1}, a\} = \{a_1, \ldots, a_{i-1}, a\}\in \M.$
\end{proof}

\cite{karp1988complexity} showed that this  algorithm has $O(\sqrt{n})$ iterations.

\begin{lemma}[\cite{karp1988complexity}]
\label{lem:roundsind}
 Algorithm~\ref{alg:3} has $O(\sqrt{n})$ steps of independence queries.
\end{lemma}

\thmindependence*
\begin{proof}
By Theorem~\ref{thm:comb}, \algoneplus \ is a $1/2-\epsilon$ approximation algorithm with $O(\log(n) \log(k))$ adaptivity if \blackbox \ satisfies the random feasibility condition, which Algorithm~\ref{alg:3} does by Lemma~\ref{lem:feasibility1}. Since there are $O(\log(n) \log(k))$ iterations of calling \blackbox \ and \blackbox \ has  $O(\sqrt{n})$ steps of independence queries by Lemma~\ref{lem:roundsind}, there are $O(\sqrt{n}\log(n)\log(k))$ total steps of independence queries.
\end{proof}

\subsection{An algorithm with $O(\log(n)\log(k))$ steps of rank queries}

\begin{lemma}
\label{lem:feasibility2}
 Algorithm~\ref{alg:rank} satisfies the random feasibility condition.
\end{lemma}
\begin{proof}
Consider $a_i$ with $i \leq \ell$. Then $a_i = b_j$ for some $j \in |N|$ such that $r_j = r_{j-1} + 1$. Since $b_1, \ldots, b_{|N|}$ is a random permutation, $b_j$ is uniformly random elements in $N \setminus \{b_1, \ldots, b_{j-1}\}$. 
We argue that $\{a : \rank(\{b_1, \ldots,  b_{j-1}, a\}) - r_{j-1} = 1\}$ is the set of all elements $a = b_\ell$ for some $\ell \geq j$ such that $\{a_1, \ldots, a_{i-1}, a\} \in \M$ by induction.

We first show that if $\rank(\{b_1, \ldots,  b_{j-1}, a\}) - r_{j-1} = 1$ then $\{a_1, \ldots, a_{i-1}, a\} \in \M$. By the algorithm, $\rank(\{b_1, \ldots,  b_{j-1}\}) = i - 1$ and by the inductive hypothesis, $\{a_1, \ldots, a_{i-1}\} \in \M$. Thus, $\{a_1, \ldots, a_{i-1}\}$ is an independent subset of $\{b_1, \ldots,  b_{j-1}\}$ of maximum size. Let $S$ be an independent subset of $\{b_1, \ldots,  b_{j-1}, a\}$ of maximum size. Since $\rank(\{b_1, \ldots,  b_{j-1}, a\}) - r_{j-1} = 1$, $|S| = |\{a_1, \ldots, a_{i-1}\}| + 1 = i$. Thus, by the augmentation property, there exists $b \in S$ such that $\{a_1, \ldots, a_{i-1}, b\} \in \M$. We have $b \neq b_{j'}$, $j' < j$ since otherwise this would contradict $\rank(\{b_1, \ldots,  b_{j-1}\}) = i - 1$. Thus $b = a$ and $\{a_1, \ldots, a_{i-1}, a\} \in \M$.

Next, we show that if $a = b_\ell$ for some $\ell \geq j$ such that $\{a_1, \ldots, a_{i-1}, a\} \in \M$, then \\ $\rank(\{b_1, \ldots,  b_{j-1}, a\}) - r_{j-1} = 1$. By the algorithm $r_{j-1} = |\{a_1, \ldots, a_{i-1}\}| = j-1$. Since $\{a_1, \ldots, a_{i-1}, a\} \in \M$,  $\rank(\{b_1, \ldots,  b_{j-1}, a\}) \geq j$. Since the rank can only increase by one when adding an element, we have  $\rank(\{b_1, \ldots,  b_{j-1}, a\}) =  j  = r_{j-1} + 1$.
\end{proof}

It is easy to see that Algorithm~\ref{alg:rank} has one step of rank queries. \cite{karp1988complexity} showed that Algorithm~\ref{alg:rank} constructs a base of $\M$.
\begin{lemma}[\cite{karp1988complexity}]
Algorithm~\ref{alg:rank} returns a base of $\M$.
\end{lemma}

\thmrank*

\begin{proof}
By Theorem~\ref{thm:comb}, \algoneplus \ is a $1/2-\O(\epsilon)$ approximation algorithm with $O(\log(n) \log(k))$ adaptivity if \blackbox \ satisfies the random feasibility condition, which Algorithm~\ref{alg:3} does by Lemma~\ref{lem:feasibility2}. Since there are $O(\log(n) \log(k))$ iterations of calling \blackbox \ and \blackbox \ has  $1$ step of rank queries, there are $O(\log(n)\log(k))$ total steps of rank queries.
\end{proof}

\end{document}